\documentclass[11pt]{article}
\usepackage{amsfonts}
\usepackage{amsmath}
\usepackage{amssymb}
\usepackage{graphicx}
\usepackage{mathtools}
\usepackage[colorlinks=true,linkcolor=blue,citecolor=magenta,plainpages=false,pdfpagelabels]
{hyperref}
\providecommand{\U}[1]{\protect\rule{.1in}{.1in}}

\newtheorem{theorem}{Theorem}

\newtheorem{definition}{Definition}

\newtheorem{lemma}[theorem]{Lemma}

\newtheorem{proposition}[theorem]{Proposition}
\newtheorem{remark}[theorem]{Remark}

\newenvironment{proof}[1][Proof]{\noindent\textbf{#1.} }{\ \rule{0.5em}{0.5em}}

\DeclareMathOperator{\Tr}{Tr}
\DeclareMathOperator{\hyp}{hyp}

\allowdisplaybreaks
\begin{document}

\title{Semi-definite optimization of the measured relative entropies of quantum
states and channels}
\author{Zixin Huang\thanks{School of Mathematical and Physical Sciences, Macquarie
University, NSW 2109, Australia} \thanks{Centre for Quantum Software and
Information, Faculty of Engineering and Information Technology, University of
Technology Sydney, Ultimo, NSW 2007, Australia}
\and Mark M. Wilde\thanks{School of Electrical and Computer Engineering, Cornell
University, Ithaca, New York 14850, USA}}
\maketitle

\begin{abstract}
The measured relative entropies of quantum states and channels find
operational significance in quantum information theory as achievable error
rates in hypothesis testing tasks. They are of interest in the near term, as
they correspond to hybrid quantum--classical strategies with technological
requirements far less challenging to implement than required by the most
general strategies allowed by quantum mechanics. In this paper, we prove that
these measured relative entropies can be calculated efficiently by means of
semi-definite programming, by making use of variational formulas for the
measured relative entropies of states and semi-definite representations of the
weighted geometric mean and the operator connection of the logarithm. Not only
do the semi-definite programs output the optimal values of the measured
relative entropies of states and channels, but they also provide numerical
characterizations of optimal strategies for achieving them, which is of
significant practical interest for designing hypothesis testing protocols.

\end{abstract}
\tableofcontents

\section{Introduction}

\subsection{Background}

In many quantum information processing tasks, distinguishing quantum states or channels is central to assessing system performance. Measures of distinguishability such as
the relative entropy \cite{Kullback1951}\ and its generalization to R\'{e}nyi
relative entropy \cite{Renyi1961OnInformation}  play important roles in information theory, finding direct operational meaning in hypothesis testing tasks
\cite{chernoff1952measure,stein_unpublished,chernoff_1956,strassen1962asymptotische,Hoeffding65,HK89} while being used to construct other entropic measures like mutual
information and conditional entropy \cite{Sibson1969,Augustin1978,Csiszar1995}. There are a number of quantum generalizations of these quantities
\cite{U62,P85,P86,muller2013quantum,wilde2014strong}, finding operational
meaning in quantum hypothesis testing tasks
\cite{hiai1991proper,ogawa2005strong,Nagaoka06,ACMBMAV07,hayashi2007error,Audenaert2008,NS09}
while being used to construct other entropic measures like quantum mutual
information and conditional entropy (see, e.g., \cite[Definition~7.18]
{khatri2024principles}). 

Typically, the optimal hypothesis-testing success rate associated with the quantum relative entropy (Umegaki relative entropy) requires collective measurements over many identical copies of a state. Implementing such joint measurements, however, demands substantial resources—potentially full-scale quantum computation \cite{Delaney2022,crossman2023quantum}.
In contrast, measured relative entropies \cite{Donald1986,P09} significantly relax these technological demands. Intuitively, the ‘measured’ versions of quantum relative entropies quantify how well two quantum states can be distinguished by
performing measurements on individual copies of the state, as opposed to using an arbitrary collective quantum measurement. Instead of collective operations, they quantify how well two states (or channels) can be distinguished when one performs a measurement on each system separately. In this way, the measured quantities capture the fundamental limits of hybrid quantum–classical discrimination strategies that are achievable with today’s experimental capabilities.
 
The measured relative entropy \cite{Donald1986,P09} has been generalized to the R\'{e}nyi family as well
\cite[Eqs.~(3.116)--(3.117)]{Fuchs1996}. Indeed, the idea behind these
measures is to evaluate classical distinguishability measures on the
distributions that result from performing a measurement on a single copy of
the state and then optimize them over all possible measurements. See
Definition~\ref{def:measured-renyi-states} and Definition~\ref{def:measured-rel-ent-states} for precise
definitions of the standard and R\'{e}nyi measured relative entropies,
respectively. In such a way, these quantities lead to technologically feasible
strategies for quantum hypothesis testing, which consist of a hybrid approach involving quantum
measurement and classical post-processing. Indeed, even though there are gaps
between the fundamental error rates of quantum hypothesis testing under
general, collective measurements and those that result from these hybrid
quantum--classical strategies, the latter strategies are more feasible in the
near term.

Beyond distinguishing states, one can also distinguish quantum channels from
one another, a task known as quantum channel discrimination, which has been
studied extensively in quantum information
\cite{CDP08,Duan09,PW09,Harrow10,MPW10,Cooney2016,WW19,wilde2020amortized,FFRS20,fang2021geometric,Bergh2024,SHW22,bergh2023infinite}. The most general strategy allowed by quantum mechanics in such a scenario is
rather complex (see \cite[Figure~1]{wilde2020amortized}), and the
technological requirements for realizing such a general strategy are even more
challenging than those needed to perform a collective measurement (i.e., one would need more complex quantum computations to realize such strategies). As such,
one can also consider relaxing the technological requirements for channel
discrimination by considering measured relative entropies of channels, as a
special case of the generalized channel divergences defined in
\cite[Definition~II.2]{Leditzky2018}. See Definition~\ref{def:measured-renyi-channels} and
Definition~\ref{def:measured-rel-ent-channels} for precise definitions of the
standard and R\'{e}nyi measured relative entropies of channels, respectively,
which also include energy constraints on the channel input state. Although the
diamond distance \cite{Kit97} and its energy-constrained counterpart
\cite{Shirokov2018,Winter2017}\ are\ in widespread use as measures of channel
distinguishability (see, e.g., \cite[Section~6.3]{khatri2024principles}), the
related general notion of measured relative entropy of channels has only been
explicitly defined more recently \cite[Eq.~(8)]{Li2022}, therein related to an
operational task called sequential channel discrimination. Here we also
explicitly define the measured R\'{e}nyi relative entropy of channels, as a
special case of the general concept from \cite[Definition~II.2]{Leditzky2018}
and \cite[Eq.~(12.12)]{SWAT18}.

\subsection{Summary of results}

In this paper, we prove that the measured relative entropies of quantum states
and channels can be computed by means of semi-definite optimization algorithms
(also known as semi-definite programs). These algorithms have runtimes that
scale efficiently with the dimension of the states and the input and output
dimensions of the channels, by employing known techniques for solving semi-definite programs \cite{PW2000,Arora2005,Arora2012,Lee2015}. Furthermore, an added benefit of these algorithms
is that, not only does one obtain the optimal values of the measured relative
entropies, but one also obtains numerically an optimal measurement for the
measured relative entropies of states and an optimal input state and
measurement for the measured relative entropies of channels. This latter
capability is of significant value for applications, in which one wishes to
construct a hybrid quantum-classical strategy for achieving the error rates of
hypothesis testing achievable by the measured relative entropies.

Our claims build upon two papers, which, coincidentally, were initially
released on the quant-ph arXiv within two days of each other
\cite{Berta2015OnEntropies,Fawzi2017}. Another edifice for our claims is
\cite{Fawzi2019}. In more detail, the paper \cite{Berta2015OnEntropies}
established variational formulas for the measured relative entropy and
measured R\'{e}nyi relative entropy, while the paper \cite{Fawzi2017} proved
that the hypograph and epigraph of the weighted geometric mean have efficient
semi-definite representations (here, see also \cite{Sagnol2013}), and the paper \cite{Fawzi2019} proved that the
hypograph of the operator connection of the logarithm has an efficient
semi-definite representation. Here, we essentially combine these findings to
arrive at our claims.

Indeed, for quantum states, our main contributions are to establish reductions
of the variational formulas of \cite{Berta2015OnEntropies} to semi-definite
optimization problems involving linear objective functions and the
aforementioned hypographs or epigraphs (see
Propositions~\ref{prop:opt-meas-renyi-states} and
\ref{prop:meas-rel-ent-states-var-rep-perspective}). This finding is
admittedly a rather direct combination of the contributions of
\cite{Berta2015OnEntropies,Fawzi2017,Fawzi2019}. However, it is ultimately
useful in establishing our next contribution, which is an extension of these
findings to measured relative entropies of channels. To
establish these latter results, we use basic properties
of weighted geometric means and the operator connection of the logarithm (see
Propositions~\ref{prop:var-rep-measured-renyi-channels}\ and
\ref{prop:var-rep-measured-div-channels}).

One benefit of our findings is that they lead to semi-definite programs involving
 linear matrix inequalities each of size $2d\times2d$ when the states are
$d\times d$ matrices and of size $2d_{A}d_{B}\times2d_{A}d_{B}$ when the channels
have input dimension $d_{A}$ and output dimension $d_{B}$ (see
Propositions~\ref{prop:opt-meas-renyi-states},
\ref{prop:meas-rel-ent-states-var-rep-perspective},
\ref{prop:var-rep-measured-renyi-channels},\ and
\ref{prop:var-rep-measured-div-channels} for precise statements). As such,
they do not suffer from the quadratic increase in size that occurs when
applying the approach from \cite{Fawzi2017,Fawzi2019} to the Petz--R\'{e}nyi
and standard quantum relative entropies (however, note that there has been
progress on addressing this issue more recently
\cite{fawzi2023optimalselfconcordantbarriersquantum,faust2023rationalapproximationsoperatormonotone,kossmann2024optimisingrelativeentropysemi}). Furthermore, it is unclear how to apply the approach from
\cite{Fawzi2017,Fawzi2019} for computing the dynamical (channel)\ version of
these quantities. However, one of our main contributions is semi-definite
programs for the measured relative entropies of channels, and the transition
from our claims for states to our claims for channels is smooth, with the
proofs consisting of just a few lines (see
\eqref{eq:proof-channel-renyi-1}--\eqref{eq:proof-channel-renyi-last}\ and~\eqref{eq:proof-channel-standard-1}--\eqref{eq:proof-channel-standard-last} for these steps).

\subsection{Organization of the paper}

The rest of our paper is organized as follows. Section~\ref{sec:notation} establishes notation and reviews background material, including the weighted
geometric mean and its properties, its hypograph and epigraph, and operator
connections and their properties (especially for the logarithm). The remaining
Sections~\ref{sec:measured-renyi-states},~\ref{sec:measured-rel-ent-states},
\ref{sec:measured-renyi-channels}, and~\ref{sec:measured-rel-ent-channels} provide essential definitions and detail our main results for measured
R\'{e}nyi relative entropy of states, measured relative entropy of states,
measured R\'{e}nyi relative entropy of channels, and measured relative entropy
of channels, respectively. We conclude in Section~\ref{sec:conclusion}\ with a
brief summary and some directions for future research.

\section{Notation and Preliminaries}

\label{sec:notation}For a Hilbert space $\mathcal{H}$, we employ the following
notation:
\[
\begin{tabular}
[c]{cl}
$\mathbb{L}(\mathcal{H})$ & set of linear operators acting on $\mathcal{H}$\\
$\mathbb{H}(\mathcal{H})$ & set of Hermitian operators acting on $\mathcal{H}
$\\
$\mathbb{P}(\mathcal{H)}$ & set of positive semi-definite operators acting on
$\mathcal{H}$\\
$\mathbb{P}_{>0}(\mathcal{H)}$ & set of positive definite operators acting on
$\mathcal{H}$\\
$\mathbb{D}(\mathcal{H)}$ & set of density operators acting on $\mathcal{H}$
\end{tabular}
\
\]
Note that $\mathbb{D}(\mathcal{H})\coloneqq\left\{  \rho\in\mathbb{P}
(\mathcal{H}):\operatorname{Tr}[\rho]=1\right\}  $.

A quantum channel is a completely positive and trace-preserving map that takes
$\mathbb{L}(\mathcal{H})$ to $\mathbb{L}(\mathcal{K})$, where $\mathcal{K}$ is
another Hilbert space. We often denote a quantum channel by $\mathcal{N}
_{A\rightarrow B}$, which indicates that the input space is $\mathbb{L}
(\mathcal{H}_{A})$ and the output space is $\mathbb{L}(\mathcal{H}_{B})$. See
\cite{Wilde2017,khatri2024principles}\ for further background on quantum
information theory.

\subsection{Weighted geometric mean and its properties}

\label{sec:weighted-geo-mean}Given positive definite operators $X,Y\in
\mathbb{P}_{>0}(\mathcal{H)}$, the weighted (operator)\ geometric mean
$X\#_{t}Y$ of weight $t\in\mathbb{R}$ is defined as
\cite{pusz1975functional,Kubo1980}
\begin{equation}
X\#_{t}Y\coloneqq X^{1/2}\left(  X^{-1/2}YX^{-1/2}\right)  ^{t}X^{1/2}.
\end{equation}
It is alternatively denoted by
\begin{equation}
G_{t}(X,Y)\coloneqq X\#_{t}Y,
\end{equation}
and we adopt this notation in what follows. The following identity holds for
all $t\in\mathbb{R}$ (see, e.g., \cite[Eq.~(7.6.5)]{khatri2024principles}):
\begin{equation}
G_{t}(X,Y)=G_{1-t}(Y,X),\label{eq:flip-identity}
\end{equation}
and so does the following identity for all $s,t\in\mathbb{R}$:
\begin{equation}
G_{s}(X,G_{t}(X,Y))=G_{st}(X,Y).
\end{equation}

The function $x\mapsto x^{t}$ is operator concave and operator monotone for
$t\in\left[  0,1\right]  $, operator antimonotone and operator convex for
$t\in\left[  -1,0\right]  $, and operator convex for $t\in\left[  1,2\right]
$ (see, e.g.,~\cite[Theorem~2.6]{carlen2010trace}). The function
\begin{equation}
\left(  X,Y\right)  \mapsto G_{t}(X,Y)
\end{equation}
is operator concave for $t\in\left[  0,1\right]  $ and operator convex for
$t\in\left[  -1,0\right]  \cup\left[  1,2\right]  $. For $t\in\left[
-1,1\right]  $, this statement is a consequence of~\cite[Theorem~3.5]
{Kubo1980} and, for $t\in\left[  1,2\right]  $, it is a consequence of
\eqref{eq:flip-identity} and~\cite[Theorem~3.5]{Kubo1980}, as well as the
aforementioned operator monotonicity properties of $x\mapsto x^{t}$. Concavity
and convexity of the function $\left(  X,Y\right)  \mapsto G_{t}(X,Y)$ is also
known as joint concavity and joint convexity of the weighted geometric mean.

A useful property of the weighted geometric mean for $t\in\left[  -1,2\right]
$\ is the transformer inequality \cite[Theorem~3.5]{Kubo1980}. For a linear operator
$K\in\mathbb{L}(\mathcal{H})$, the following inequality holds for all
$t\in\left[  0,1\right]  $:
\begin{equation}
KG_{t}(X,Y)K^{\dag}\leq G_{t}(KXK^{\dag},KYK^{\dag}),\label{eq:KA-transformer}
\end{equation}
and the opposite inequality holds for all $t\in\left[  -1,0\right]
\cup\left[  1,2\right]  $:
\begin{equation}
KG_{t}(X,Y)K^{\dag}\geq G_{t}(KXK^{\dag},KYK^{\dag}
).\label{eq:KA-transformer-2}
\end{equation}
These inequalities are saturated when $K$ is invertible; i.e., for all
$t\in\left[  -1,2\right]  $ and invertible $K$, the following holds:
\begin{equation}
KG_{t}(X,Y)K^{\dag}=G_{t}(KXK^{\dag},KYK^{\dag}).\label{eq:KA-transformer-eq}
\end{equation}
The inequalities in~\eqref{eq:KA-transformer}--\eqref{eq:KA-transformer-2}
were proven for all $t\in\left[  -1,1\right]  $ in~\cite[Theorem~3.5]
{Kubo1980}, and the extension to $t\in\left[  1,2\right]  $ follows from
\eqref{eq:flip-identity} and~\cite[Theorem~3.5]{Kubo1980}. See also
\cite[Lemma~47]{fang2021geometric}.

\subsection{Hypograph and epigraph of the weighted geometric mean}

\label{sec:hypo-epi}For $t\in\left[  0,1\right]  $, the operator hypograph of
$G_{t}$ is given by~\cite[Section~3.1]{Fawzi2017}
\begin{equation}
\text{hyp}_{t}\coloneqq\left\{  \left(  X,Y,T\right)  \in\mathbb{P}
_{>0}(\mathcal{H)\times}\mathbb{P}_{>0}(\mathcal{H)\times}\mathbb{H}
(\mathcal{H}):G_{t}(X,Y)\geq T\right\}  ,
\end{equation}
and for $t\in\left[  -1,0\right]  \cup\left[  1,2\right]  $, the operator
epigraph of $G_{t}$ is given by~\cite[Section~3.1]{Fawzi2017}
\begin{equation}
\text{epi}_{t}\coloneqq\left\{  \left(  X,Y,T\right)  \in\mathbb{P}
_{>0}(\mathcal{H)\times}\mathbb{P}_{>0}(\mathcal{H)\times}\mathbb{H}
(\mathcal{H}):G_{t}(X,Y)\leq T\right\}  .
\end{equation}
These sets are convex due to the aforementioned concavity and convexity
properties of $G_{t}$.

As a consequence of~\cite[Theorem~3]{Fawzi2017}, for all rational $t\in\left[
0,1\right]  $, the set hyp$_{t}$ is semi-definite representable (see also \cite{Sagnol2013}), and for all
rational $t\in\left[  -1,0\right]  \cup\left[  1,2\right]  $, the set
epi$_{t}$ is semi-definite representable. This means that these sets can be
represented in terms of a finite number of linear matrix inequalities
\cite{Nie2012} and implies that one can use the methods of semi-definite
programming to optimize over elements of these sets. This fact was put to use
in~\cite{fang2021geometric,Fawzi2021definingquantum} for quantum
information-theoretic applications, and we make use of it here as well.

\subsection{Operator connections}

\label{sec:op-connection}Generalizing the notion of an operator geometric
mean, an operator connection is defined in terms of an operator monotone
function $f$ as~\cite{Kubo1980}
\begin{equation}
P_{f}(X,Y)\coloneqq X^{1/2}f\!\left(  X^{-1/2}YX^{-1/2}\right)  X^{1/2},
\end{equation}
where $X,Y\in\mathbb{P}_{>0}(\mathcal{H)}$. This is also known as a
non-commutative perspective function~\cite{Effros2009,Ebadian2011,Effros2014}.
Due to~\cite[Theorem~3.5]{Kubo1980}, the function
\begin{equation}
\left(  X,Y\right)  \mapsto P_{f}(X,Y)
\end{equation}
is operator concave (i.e., jointly concave), and the transformer inequality
holds for every linear operator $K\in\mathbb{L}(\mathcal{H})$:
\begin{equation}
KP_{f}(X,Y)K^{\dag}\leq P_{f}(KXK^{\dag},KYK^{\dag}
).\label{eq:transformer-connection}
\end{equation}
Equality holds in~\eqref{eq:transformer-connection} if $K$ is invertible;
i.e., for invertible $K\in\mathbb{L}(\mathcal{H})$, the following equality
holds:
\begin{equation}
KP_{f}(X,Y)K^{\dag}=P_{f}(KXK^{\dag},KYK^{\dag}
).\label{eq:transformer-conn-eq}
\end{equation}
The operator hypograph of $P_{f}$ is given by
\begin{equation}
\text{hyp}_{f}\coloneqq\left\{  \left(  X,Y,T\right)  \in\mathbb{P}
_{>0}(\mathcal{H)\times}\mathbb{P}_{>0}(\mathcal{H)\times}\mathbb{H}
(\mathcal{H}):P_{f}(X,Y)\geq T\right\}  ,
\end{equation}
and it is a convex set due to the aforementioned joint concavity of
$P_{f}(X,Y)$.

The logarithm is the main example of an operator monotone function on which we
focus, other than the power functions from Section~\ref{sec:weighted-geo-mean}, due to its connection with relative entropy. Furthermore, there is an
efficient semi-definite approximation of the hypograph of the connection of
the logarithm (i.e., hyp$_{\ln}$)~\cite[Theorem~3]{Fawzi2019}, which leads to
semi-definite optimization algorithms for calculating measured relative
entropies of states and channels. This fact was put to use in
\cite{Fawzi2018,Wilde2018,bunandar2020numerical,Coutts2021certifying,Wang_2024}
for quantum information-theoretic applications, and we make use of it here as
well. To be clear, we use the following notation later on:
\begin{equation}
P_{\ln}(X,Y)\coloneqq X^{1/2}\ln\!\left(  X^{-1/2}YX^{-1/2}\right)
X^{1/2}.\label{eq:log-perspect}
\end{equation}
This is also related to the operator relative entropy~\cite{fujii1989relative}, for which one finds the following notation in the literature
\cite{fang2021geometric}:
\begin{align}
D_{\text{op}}(X\Vert Y) &  \coloneqq-P_{\ln}(X,Y)\\
&  =X^{1/2}\ln\!\left(  X^{1/2}Y^{-1}X^{1/2}\right)  X^{1/2}.
\end{align}

\section{Measured R\'{e}nyi relative entropy of states}

\label{sec:measured-renyi-states}

\subsection{Definition and basic properties}

Given a probability distribution $p\equiv\left(  p(x)\right)  _{x\in
\mathcal{X}}$ and a non-negative function $q\equiv\left(  q(x)\right)
_{x\in\mathcal{X}}$, the R\'{e}nyi relative entropy is defined for $\alpha
\in\left(  0,1\right)  \cup(1,\infty)$ as~\cite{Renyi1961OnInformation}
\begin{equation}
D_{\alpha}(p\Vert q)\coloneqq\frac{1}{\alpha-1}\ln\sum_{x\in\mathcal{X}
}p(x)^{\alpha}q(x)^{1-\alpha}, \label{eq:renyi-rel-ent-def}
\end{equation}
when $\alpha\in\left(  0,1\right)  $ or when $\alpha>1$ and
$\operatorname{supp}(p)\subseteq\operatorname{supp}(q)$. Otherwise, when
$\alpha>1$ and $\operatorname{supp}(p)\not \subseteq \operatorname{supp}(q)$,
it is set to $+\infty$. The R\'{e}nyi relative entropy satisfies the
data-processing inequality for all $\alpha\in\left(  0,1\right)  \cup
(1,\infty)$, which means that
\begin{equation}
D_{\alpha}(p\Vert q)\geq D_{\alpha}(N(p)\Vert N(q)),
\label{eq:DP-classical-Renyi}
\end{equation}
where $N$ is a classical channel (i.e., a conditional probability distribution
with elements $\left(  N(y|x)\right)  _{y\in\mathcal{Y},x\in\mathcal{X}}$) and
the notation $N(p)$ is a shorthand for the distribution that results from
processing $p$ with $N$:
\begin{equation}
N(p)\equiv\left(  \sum_{x\in\mathcal{X}}N(y|x)p(x)\right)  _{y\in\mathcal{Y}},
\label{eq:classical-channel}
\end{equation}
with a similar meaning for $N(q)$. The R\'{e}nyi relative entropy is also
monotone in the parameter $\alpha$; i.e., for $\beta>\alpha>0$, the following
inequality holds for all $p$ and $q$:
\begin{equation}
D_{\alpha}(p\Vert q)\leq D_{\beta}(p\Vert q). \label{eq:renyi-ordered}
\end{equation}
See~\cite{Erven2014} for a review of the classical R\'{e}nyi relative entropy in~\eqref{eq:renyi-rel-ent-def}.

\begin{definition}[Measured R\'{e}nyi relative entropy]
\label{def:measured-renyi-states}
Given a quantum state~$\rho$ and a positive semi-definite operator $\sigma$,
the measured R\'{e}nyi relative entropy is defined by optimizing the R\'{e}nyi
relative entropy over all possible measurements~\cite[Eqs.~(3.116)--(3.117)]
{Fuchs1996}:
\begin{equation}
D_{\alpha}^{M}(\rho\Vert\sigma)\coloneqq\sup_{\mathcal{X},\left(  \Lambda
_{x}\right)  _{x\in\mathcal{X}}}\frac{1}{\alpha-1}\ln\sum_{x\in\mathcal{X}
}\operatorname{Tr}[\Lambda_{x}\rho]^{\alpha}\operatorname{Tr}[\Lambda
_{x}\sigma]^{1-\alpha}, \label{eq:measured-renyi-def}
\end{equation}
where the supremum is over every finite alphabet $\mathcal{X}$ and every
positive operator-valued measure (POVM) $\left(  \Lambda_{x}\right)
_{x\in\mathcal{X}}$ (i.e., satisfying $\Lambda_{x}\geq0$ for all
$x\in\mathcal{X}$ and $\sum_{x\in\mathcal{X}}\Lambda_{x}=I$).
\end{definition}

For $\alpha>1$,
the measured R\'{e}nyi relative entropy is finite if and only if
$\operatorname{supp}(\rho)\subseteq\operatorname{supp}(\sigma)$. If the
support condition holds, then it follows that the support of $\left(
\operatorname{Tr}[\Lambda_{x}\rho]\right)  _{x\in\mathcal{X}}$ is contained in
the support of $\left(  \operatorname{Tr}[\Lambda_{x}\sigma]\right)
_{x\in\mathcal{X}}$, which in turn implies that $D_{\alpha}^{M}(\rho
\Vert\sigma)<+\infty$. If the support condition does not hold, then
$D_{\alpha}^{M}(\rho\Vert\sigma)=+\infty$, by applying the argument in
\cite[Section~3.5]{rippchen2024locally}.

We now recall some basic properties of the measured R\'{e}nyi relative
entropies, the first of which is actually a consequence of~\cite[Theorem~4]
{Berta2015OnEntropies}\ and the second observed in~\cite[Lemma~5]
{rippchen2024locally}.

\begin{proposition}
\label{prop:rank-one-opt}It suffices to optimize $D_{\alpha}^{M}(\rho
\Vert\sigma)$ over rank-one POVMs; i.e.,
\begin{equation}
D_{\alpha}^{M}(\rho\Vert\sigma)=\sup_{\mathcal{X},\left(  \varphi_{x}\right)
_{x\in\mathcal{X}}}\frac{1}{\alpha-1}\ln\sum_{x\in\mathcal{X}}
\operatorname{Tr}[\varphi_{x}\rho]^{\alpha}\operatorname{Tr}[\varphi_{x}
\sigma]^{1-\alpha},
\end{equation}
where each $\varphi_{x}$ is a rank-one operator such that $\sum_{x\in
\mathcal{X}}\varphi_{x}=I$.
\end{proposition}

\begin{proof}
This is a direct consequence of the data-processing inequality in~\eqref{eq:DP-classical-Renyi}. Indeed, by diagonalizing $\Lambda_{x}$ as
$\Lambda_{x}=\sum_{z\in\mathcal{Z}}\phi_{x,z}$, where each $\phi_{x,z}$ is
rank one, consider that every POVM\ $\left(  \Lambda_{x}\right)
_{x\in\mathcal{X}}$ can be understood as a coarse graining of the
POVM\ $\left(  \phi_{x,z}\right)  _{x\in\mathcal{X},z\in\mathcal{Z}}$ because
\begin{equation}
\operatorname{Tr}[\Lambda_{x}\rho]=\sum_{z\in\mathcal{Z}}\operatorname{Tr}
[\phi_{x,z}\rho].
\end{equation}
By defining $p_{X,Z}(x,z)\coloneqq \operatorname{Tr}[\phi_{x,z}\rho]$ and
$q_{X,Z}(x,z)\coloneqq \operatorname{Tr}[\phi_{x,z}\sigma]$ and noting that
one obtains $p_{X}(x)=\operatorname{Tr}[\Lambda_{x}\rho]$ and $q_{X}
(x)=\operatorname{Tr}[\Lambda_{x}\sigma]$ by marginalization (a particular
kind of classical channel), the data-processing inequality in~\eqref{eq:DP-classical-Renyi} implies that
\begin{equation}
D_{\alpha}(p_{X,Z}\Vert q_{X,Z})\geq D_{\alpha}(p_{X}\Vert q_{X}),
\end{equation}
concluding the proof.
\end{proof}

\begin{proposition}
\label{prop:DP-measured-renyi}The measured R\'{e}nyi relative entropy obeys
the data-processing inequality; i.e., for every state $\rho$, positive
semi-definite operator $\sigma$, quantum channel $\mathcal{N}$, and $\alpha
\in\left(  0,1\right)  \cup\left(  1,\infty\right)  $, the following
inequality holds:
\begin{equation}
D_{\alpha}^{M}(\rho\Vert\sigma)\geq D_{\alpha}^{M}(\mathcal{N}(\rho
)\Vert\mathcal{N}(\sigma)). \label{eq:measured-Renyi-DP}
\end{equation}

\end{proposition}

\begin{proof}
Observe that
\begin{align}
&  \frac{1}{\alpha-1}\ln\sum_{x\in\mathcal{X}}\operatorname{Tr}[\Lambda
_{x}\mathcal{N}(\rho)]^{\alpha}\operatorname{Tr}[\Lambda_{x}\mathcal{N}
(\sigma)]^{1-\alpha}\nonumber\\
&  =\frac{1}{\alpha-1}\ln\sum_{x\in\mathcal{X}}\operatorname{Tr}
[\mathcal{N}^{\dag}(\Lambda_{x})\rho]^{\alpha}\operatorname{Tr}[\mathcal{N}
^{\dag}(\Lambda_{x})\sigma]^{1-\alpha}\\
&  \leq D_{\alpha}^{M}(\rho\Vert\sigma).
\end{align}
In the above, we made use of the Hilbert--Schmidt adjoint $\mathcal{N}^{\dag}
$, which is completely positive and unital, implying that $\left(
\mathcal{N}^{\dag}(\Lambda_{x})\right)  _{x\in\mathcal{X}}$ is a POVM. The
inequality follows from the fact that $D_{\alpha}^{M}(\rho\Vert\sigma)$
involves an optimization over every alphabet $\mathcal{X}$ and POVM. Since the
inequality holds for every POVM~$\left(  \Lambda_{x}\right)  _{x\in
\mathcal{X}}$, we conclude~\eqref{eq:measured-Renyi-DP}. (Here we can also
observe that the claim holds more generally for positive, trace-preserving maps.)
\end{proof}

\bigskip

It can be helpful to write the measured R\'{e}nyi relative entropy in terms of
the measured R\'{e}nyi relative quasi-entropy:
\begin{equation}
D_{\alpha}^{M}(\rho\Vert\sigma)=\frac{1}{\alpha-1}\ln Q_{\alpha}^{M}(\rho
\Vert\sigma),
\end{equation}
where the latter is defined as
\begin{equation}
Q_{\alpha}^{M}(\rho\Vert\sigma)\coloneqq\left\{
\begin{array}
[c]{cc}
\inf_{\mathcal{X},\left(  \Lambda_{x}\right)  _{x\in\mathcal{X}}}\sum
_{x\in\mathcal{X}}\operatorname{Tr}[\Lambda_{x}\rho]^{\alpha}\operatorname{Tr}
[\Lambda_{x}\sigma]^{1-\alpha} & \text{for }\alpha\in\left(  0,1\right) \\
\sup_{\mathcal{X},\left(  \Lambda_{x}\right)  _{x\in\mathcal{X}}}\sum
_{x\in\mathcal{X}}\operatorname{Tr}[\Lambda_{x}\rho]^{\alpha}\operatorname{Tr}
[\Lambda_{x}\sigma]^{1-\alpha} & \text{for }\alpha>1
\end{array}
\right.  . \label{eq:quasi-renyi}
\end{equation}

One can also define the projectively measured R\'{e}nyi relative entropy as
\begin{equation}
D_{\alpha}^{P}(\rho\Vert\sigma)\coloneqq\frac{1}{\alpha-1}\ln Q_{\alpha}
^{P}(\rho\Vert\sigma),
\end{equation}
where
\begin{equation}
Q_{\alpha}^{P}(\rho\Vert\sigma)\coloneqq\left\{
\begin{array}
[c]{cc}
\inf_{\left(  \Pi_{x}\right)  _{x\in\mathcal{X}}}\sum
_{x\in\mathcal{X}}\operatorname{Tr}[\Pi_{x}\rho]^{\alpha}\operatorname{Tr}
[\Pi_{x}\sigma]^{1-\alpha} & \text{for }\alpha\in\left(  0,1\right) \\
\sup_{\left(  \Pi_{x}\right)  _{x\in\mathcal{X}}}\sum
_{x\in\mathcal{X}}\operatorname{Tr}[\Pi_{x}\rho]^{\alpha}\operatorname{Tr}
[\Pi_{x}\sigma]^{1-\alpha} & \text{for }\alpha>1
\end{array}
\right.  ,
\end{equation}
with the key difference being that the optimization is performed over every
projective measurement $\left(  \Pi_{x}\right)  _{x\in\mathcal{X}}$ (i.e.,
satisfying $\Pi_{x}\Pi_{x^{\prime}}=\Pi_{x}\delta_{x,x^{\prime}}$ for all
$x,x^{\prime}\in\mathcal{X}$ in addition to the requirements of a POVM) and the size of the alphabet $\mathcal{X}$ is equal to the dimension of the underlying Hilbert space of $\rho$ and $\sigma$.

It is known from~\cite[Theorem~4]{Berta2015OnEntropies} that the following
equalities hold for all $\alpha\in\left(  0,1\right)  \cup\left(
1,\infty\right)  $:
\begin{equation}
D_{\alpha}^{M}(\rho\Vert\sigma)=D_{\alpha}^{P}(\rho\Vert\sigma),\qquad
Q_{\alpha}^{M}(\rho\Vert\sigma)=Q_{\alpha}^{P}(\rho\Vert\sigma
),\label{eq:proj-equals-measured}
\end{equation}
which is a non-trivial finding that makes use of operator concavity and
convexity properties of the function $x\mapsto x^{t}$. Furthemore, it was
noted therein that the measured R\'{e}nyi relative entropy is achieved by a
rank-one, projective measurement. This has practical implications for
achieving the measured R\'{e}nyi relative entropy because projective
measurements are simpler to realize experimentally than general POVMs.

\subsection{Variational formulas for measured R\'{e}nyi relative entropy of
states}

As a consequence of~\cite[Lemma~3 and Theorem~4]{Berta2015OnEntropies}, the
measured R\'{e}nyi relative entropy has the following variational formulas for
all $\alpha\in\left(  0,1\right)  \cup\left(  1,\infty\right)  $:
\begin{align}
D_{\alpha}^{M}(\rho\Vert\sigma) &  =\sup_{\omega>0}\left\{  \frac{1}{\alpha
-1}\ln\!\left(  \alpha\operatorname{Tr}[\omega\rho]+\left(  1-\alpha\right)
\operatorname{Tr}[\omega^{\frac{\alpha}{\alpha-1}}\sigma]\right)  \right\}
\label{eq:measured-Renyi-var-add}\\
&  =\sup_{\omega>0}\left\{  \frac{1}{\alpha-1}\ln\!\left(  \left(
\operatorname{Tr}[\omega\rho]\right)  ^{\alpha}\left(  \operatorname{Tr}
[\omega^{\frac{\alpha}{\alpha-1}}\sigma]\right)  ^{1-\alpha}\right)  \right\}
.\label{eq:measured-Renyi-var-mult}
\end{align}
These are a direct consequence of and equivalent to the precise expressions
given in~\cite[Lemma~3]{Berta2015OnEntropies}, which are as follows:
\begin{align}
Q_{\alpha}^{M}(\rho\Vert\sigma) &  =\left\{
\begin{array}
[c]{cc}
\inf_{\omega>0}\left\{  \alpha\operatorname{Tr}[\omega\rho]+\left(
1-\alpha\right)  \operatorname{Tr}[\omega^{\frac{\alpha}{\alpha-1}}
\sigma]\right\}   & \text{for }\alpha\in\left(  0,1/2\right)  \\
\inf_{\omega>0}\left\{  \alpha\operatorname{Tr}[\omega^{1-\frac{1}{\alpha}
}\rho]+\left(  1-\alpha\right)  \operatorname{Tr}[\omega\sigma]\right\}   &
\text{for }\alpha\in\lbrack1/2,1)\\
\sup_{\omega>0}\left\{  \alpha\operatorname{Tr}[\omega^{1-\frac{1}{\alpha}
}\rho]+\left(  1-\alpha\right)  \operatorname{Tr}[\omega\sigma]\right\}   &
\text{for }\alpha>1
\end{array}
\right.  ,\label{eq:quasi-measured-precise}\\
Q_{\alpha}^{M}(\rho\Vert\sigma) &  =\left\{
\begin{array}
[c]{cc}
\inf_{\omega>0}\left\{  \left(  \operatorname{Tr}[\omega\rho]\right)
^{\alpha}\left(  \operatorname{Tr}[\omega^{\frac{\alpha}{\alpha-1}}
\sigma]\right)  ^{1-\alpha}\right\}   & \text{for }\alpha\in\left(
0,1\right)  \\
\sup_{\omega>0}\left\{  \left(  \operatorname{Tr}[\omega^{1-\frac{1}{\alpha}
}\rho]\right)  ^{\alpha}\left(  \operatorname{Tr}[\omega\sigma]\right)
^{1-\alpha}\right\}   & \text{for }\alpha>1
\end{array}
\right.  .\label{eq:quasi-measured-precise-mult}
\end{align}
Indeed, one obtains~\eqref{eq:measured-Renyi-var-add} for $\alpha\geq1/2$ from
the second and third expressions in~\eqref{eq:quasi-measured-precise} by the
substitution $\omega\rightarrow\omega^{\frac{\alpha}{\alpha-1}}$, and
similarly for getting~\eqref{eq:measured-Renyi-var-mult} from
\eqref{eq:quasi-measured-precise-mult} for $\alpha>1$. We note in passing that
these variational formulas have found application in devising variational
quantum algorithms for estimating the measured R\'{e}nyi relative
entropy~\cite{Goldfeld2024}.

Although the expressions in
\eqref{eq:measured-Renyi-var-add}--\eqref{eq:measured-Renyi-var-mult} are
simpler than those in
\eqref{eq:quasi-measured-precise}--\eqref{eq:quasi-measured-precise-mult}, the
various expressions in~\eqref{eq:quasi-measured-precise} are helpful for
seeing that the optimizations can be performed efficiently (see
Proposition~\ref{prop:opt-meas-renyi-states}\ for further details). To see
this, let us consider the expressions in~\eqref{eq:quasi-measured-precise} one
at a time. For $\alpha\in\left(  0,1/2\right)  $, the function $\omega
\mapsto\omega^{\frac{\alpha}{\alpha-1}}$ is operator convex because
$\frac{\alpha}{\alpha-1}\in\left(  -1,0\right)  $. As such, the objective
function $\alpha\operatorname{Tr}[\omega\rho]+\left(  1-\alpha\right)
\operatorname{Tr}[\omega^{\frac{\alpha}{\alpha-1}}\sigma]$ is convex in
$\omega$. For $\alpha\in\lbrack1/2,1)$, the function $\omega\mapsto
\omega^{1-\frac{1}{\alpha}}$ is operator convex because $1-\frac{1}{\alpha}
\in\lbrack-1,0)$. Then the objective function $\alpha\operatorname{Tr}
[\omega^{1-\frac{1}{\alpha}}\rho]+\left(  1-\alpha\right)  \operatorname{Tr}
[\omega\sigma]$ is convex in $\omega$. Finally, for $\alpha>1$, the function
$\omega\mapsto\omega^{1-\frac{1}{\alpha}}$ is operator concave because
$1-\frac{1}{\alpha}\in\left(  0,1\right)  $. Then the objective function
$\alpha\operatorname{Tr}[\omega^{1-\frac{1}{\alpha}}\rho]+\left(
1-\alpha\right)  \operatorname{Tr}[\omega\sigma]$ is concave in $\omega$.

An operator $\omega$ that achieves the optimal values in
\eqref{eq:quasi-measured-precise}--\eqref{eq:quasi-measured-precise-mult}
corresponds to an observable whose eigenvectors form an optimal measurement
for achieving the measured R\'{e}nyi relative entropy. This point becomes
clear by inspecting the proof of
Proposition~\ref{prop:alt-proof-var-measured-Renyi} below. As such, being able
to calculate such an observable numerically is valuable from an operational
perspective, and we note here that this task is accomplished by the
semi-definite optimization algorithm mentioned in
Proposition~\ref{prop:opt-meas-renyi-states}. It has been known for some time
that the optimal observable for $\alpha=1/2$ has an analytical form~\cite{Fuchs1995}, given by $G_{\frac{1}{2}}(\sigma^{-1},\rho)$ and known as
the Fuchs--Caves observable (see also~\cite[Section~3.3]{Fuchs1996}).

In Appendix~\ref{app:alt-proof-measured-renyi}, we provide an alternative proof of~\eqref{eq:measured-Renyi-var-add},
which makes use of the inequality of arithmetic and geometric means, as well as
Bernoulli's inequality. We think this proof is of interest due to its
simplicity. Let us note that the expressions in
\eqref{eq:quasi-measured-precise-mult} were actually shown in the proof of
\cite[Lemma~3]{Berta2015OnEntropies} to follow from
\eqref{eq:quasi-measured-precise} by means of these inequalities.

\subsection{Optimizing the measured R\'{e}nyi relative entropy of states}

One of the main goals of~\cite{Berta2015OnEntropies} was to derive variational
formulas for the measured R\'{e}nyi relative entropy and explore applications
of them in quantum information. In this section, we observe in
Proposition~\ref{prop:opt-meas-renyi-states} below that there is an alternative
variational representation of the measured R\'{e}nyi relative entropy in terms
of a linear objective function and the hypograph or epigraph of the weighted
geometric mean. From this observation, we conclude that there is an efficient
semi-definite optimization algorithm for computing the measured R\'{e}nyi
relative entropy, which makes use of the fact recalled in
Section~\ref{sec:hypo-epi} (i.e., from \cite[Theorem~3]{Fawzi2017}). As
mentioned previously, not only does this algorithm compute the optimal value
of $Q_{\alpha}^{M}(\rho\Vert\sigma)$ for all $\alpha\in\left(  0,1\right)
\cup\left(  1,\infty\right)  $, but it also determines an optimal observable
$\omega$. Another advantage of the variational representations in
Proposition~\ref{prop:opt-meas-renyi-states}\ is that they lead to a rather
rapid derivation of variational representations of the measured R\'{e}nyi
relative entropy of channels (see
Proposition~\ref{prop:var-rep-measured-renyi-channels}).

\begin{proposition}
\label{prop:opt-meas-renyi-states}Let $\rho$ be a state and $\sigma$ a
positive semi-definite operator. For $\alpha\in\left(  0,1/2\right)  $,
\begin{equation}
Q_{\alpha}^{M}(\rho\Vert\sigma)=\inf_{\omega,\theta>0}\left\{  \alpha
\operatorname{Tr}[\omega\rho]+\left(  1-\alpha\right)  \operatorname{Tr}
[\theta\sigma]:\theta\geq G_{\frac{\alpha}{\alpha-1}}(I,\omega)\right\}
,\label{eq:meas-renyi-states-a-0-half}
\end{equation}
for $\alpha\in\lbrack1/2,1)$,
\begin{equation}
Q_{\alpha}^{M}(\rho\Vert\sigma)=\inf_{\omega,\theta>0}\left\{  \alpha
\operatorname{Tr}[\theta\rho]+\left(  1-\alpha\right)  \operatorname{Tr}
[\omega\sigma]:\theta\geq G_{1-\frac{1}{\alpha}}(I,\omega)\right\}
,\label{eq:meas-renyi-states-a-half-1}
\end{equation}
and for $\alpha>1$,
\begin{equation}
Q_{\alpha}^{M}(\rho\Vert\sigma)=\sup_{\omega,\theta>0}\left\{  \alpha
\operatorname{Tr}[\theta\rho]+\left(  1-\alpha\right)  \operatorname{Tr}
[\omega\sigma]:\theta\leq G_{1-\frac{1}{\alpha}}(I,\omega)\right\}
.\label{eq:meas-renyi-states-a-gt-1}
\end{equation}

For all rational $\alpha\in\left(  0,1\right)  \cup\left(  1,\infty\right)  $,
the quantity $Q_{\alpha}^{M}(\rho\Vert\sigma)$ can be calculated by means of a
semi-definite program. More specifically, when $\rho$ and $\sigma$ are
$d\times d$ matrices and $p$ and $q$ are relatively prime integers such that
$\frac{p}{q}=\frac{\alpha}{\alpha-1}$ for $\alpha\in\left(  0,1/2\right)  $ or
$\frac{p}{q}=1-\frac{1}{\alpha}$ for $\alpha\in\left[  1/2,1\right)
\cup(1,\infty)$, the semi-definite program requires $O(\log_{2}q)$ linear
matrix inequalities each of size $2d\times2d$.
\end{proposition}

\begin{proof}
These formulas are a direct consequence of~\eqref{eq:quasi-measured-precise}
and the following identities:
\begin{equation}
\omega^{\frac{\alpha}{\alpha-1}}=G_{\frac{\alpha}{\alpha-1}}(I,\omega
),\qquad\omega^{1-\frac{1}{\alpha}}=G_{1-\frac{1}{\alpha}}(I,\omega
),\label{eq:rewrite-op-geo-mean}
\end{equation}
while noting that the optimal value of $\theta$ in
\eqref{eq:meas-renyi-states-a-0-half}\ is equal to $G_{\frac{\alpha}{\alpha
-1}}(I,\omega)$ and the optimal value of $\theta$ in
\eqref{eq:meas-renyi-states-a-half-1} and~\eqref{eq:meas-renyi-states-a-gt-1}
is equal to $G_{1-\frac{1}{\alpha}}(I,\omega)$.

As such, we have rewritten $Q_{\alpha}^{M}(\rho\Vert\sigma)$ for $\alpha
\in\left(  0,1/2\right)  $ in terms of the hypograph of $G_{\frac{\alpha
}{\alpha-1}}$, for $\alpha\in\lbrack1/2,1)$ in terms of the hypograph of
$G_{1-\frac{1}{\alpha}}$, and for $\alpha>1$ in terms of the epigraph of
$G_{1-\frac{1}{\alpha}}$. By appealing to~\cite[Theorem~3]{Fawzi2017}, it
follows that all of these quantities can be efficiently calculated for
rational $\alpha$ by means of semi-definite programming, with complexity as
stated above.
\end{proof}

\section{Measured relative entropy of states}

\label{sec:measured-rel-ent-states}

\subsection{Definition and basic properties}

Given a probability distribution $p\equiv\left(  p(x)\right)  _{x\in
\mathcal{X}}$ and a non-negative function $q\equiv\left(  q(x)\right)
_{x\in\mathcal{X}}$, the relative entropy is defined as
\begin{equation}
D(p\Vert q)\coloneqq\sum_{x\in\mathcal{X}}p(x)\ln\!\left(  \frac{p(x)}
{q(x)}\right)  ,
\end{equation}
when $\operatorname{supp}(p)\subseteq\operatorname{supp}(q)$ and it is set to
$+\infty$ otherwise. It is equal to the $\alpha\rightarrow1$ limit of the
R\'{e}nyi relative entropy in~\eqref{eq:renyi-rel-ent-def}:
\begin{equation}
D(p\Vert q)=\lim_{\alpha\rightarrow1}D_{\alpha}(p\Vert q).
\label{eq:alpha-1-limit}
\end{equation}
By virtue of the ordering property in~\eqref{eq:renyi-ordered}, we can write
\begin{equation}
D(p\Vert q)=\sup_{\alpha\in\left(  0,1\right)  }D_{\alpha}(p\Vert
q)=\inf_{\alpha>1}D_{\alpha}(p\Vert q).
\end{equation}
As a direct consequence of~\eqref{eq:DP-classical-Renyi} and
\eqref{eq:alpha-1-limit}, the relative entropy satisfies the data-processing
inequality, which means that
\begin{equation}
D(p\Vert q)\geq D(N(p)\Vert N(q)),
\end{equation}
where we used the same notation from~\eqref{eq:classical-channel}.

\begin{definition}[Measured relative entropy]
\label{def:measured-rel-ent-states}
Given a quantum state $\rho$ and a positive semi-definite operator $\sigma$,
the measured relative entropy is defined by optimizing the relative entropy
over all possible measurements~\cite{Donald1986,P09}:
\begin{equation}
D^{M}(\rho\Vert\sigma)\coloneqq\sup_{\mathcal{X},\left(  \Lambda_{x}\right)
_{x\in\mathcal{X}}}\sum_{x\in\mathcal{X}}\operatorname{Tr}[\Lambda_{x}\rho
]\ln\!\left(  \frac{\operatorname{Tr}[\Lambda_{x}\rho]}{\operatorname{Tr}
[\Lambda_{x}\sigma]}\right)  , \label{eq:measured-rel-ent-def}
\end{equation}
where the supremum is over every finite alphabet $\mathcal{X}$ and every
positive operator-valued measure (POVM) $\left(  \Lambda_{x}\right)
_{x\in\mathcal{X}}$ (i.e., satisfying $\Lambda_{x}\geq0$ for all
$x\in\mathcal{X}$ and $\sum_{x\in\mathcal{X}}\Lambda_{x}=I$).
\end{definition}

\begin{proposition}
\label{prop:alpha-1-limit-measured}
For $\rho$ a state and $\sigma$ a positive semi-definite operator, the
measured relative entropy is equal to the $\alpha\rightarrow1$ limit of the
measured R\'{e}nyi relative entropy:
\begin{equation}
D^{M}(\rho\Vert\sigma)=\lim_{\alpha\rightarrow1}D_{\alpha}^{M}(\rho\Vert
\sigma).
\end{equation}

\end{proposition}

\begin{proof}
See Appendix~\ref{app:alpha-1-limit}.
\end{proof}

\medskip

Let us note that the convergence statement in Proposition~\ref{prop:alpha-1-limit-measured} can be made more precise by
invoking~\eqref{eq:renyi-ordered}\ and \cite[Lemma~23]
{cheng2024invitationsamplecomplexityquantum} (see also \cite[Lemma~13]
{Bergh2024}). Specifically, when $\operatorname{supp}(\rho)\subseteq
\operatorname{supp}(\sigma)$, there exists a state-dependent constant
$c(\rho\Vert\sigma)$ such that, for all $\delta\in(0,\frac{\ln3}{2c(\rho
\Vert\sigma)}]$, the following bound holds:
\begin{equation}
D_{1-\delta}^{M}(\rho\Vert\sigma)\leq D^{M}(\rho\Vert\sigma)\leq D_{1-\delta
}^{M}(\rho\Vert\sigma)+\delta K\left[  c(\rho\Vert\sigma)\right]
^{2},\label{eq:bound-performance-renyi-approach}
\end{equation}
where $K\coloneqq\cosh(\left(  \ln3\right)  /2)$. We invoke this bound later
on in Remark~\ref{rem:alt-approach}.

The following statements can be proven similarly to
Propositions~\ref{prop:rank-one-opt}\ and~\ref{prop:DP-measured-renyi}, or
alternatively, they can be understood as the $\alpha\rightarrow1$ limit of
these propositions.

\begin{proposition}
It suffices to optimize $D^{M}(\rho\Vert\sigma)$ over rank-one POVMs; i.e.,
\begin{equation}
D^{M}(\rho\Vert\sigma)=\sup_{\mathcal{X},\left(  \varphi_{x}\right)
_{x\in\mathcal{X}}}\sum_{x\in\mathcal{X}}\operatorname{Tr}[\varphi_{x}\rho
]\ln\!\left(  \frac{\operatorname{Tr}[\varphi_{x}\rho]}{\operatorname{Tr}
[\varphi_{x}\sigma]}\right)  ,
\end{equation}
where each $\varphi_{x}$ is a rank-one operator such that $\sum_{x\in
\mathcal{X}}\varphi_{x}=I$.
\end{proposition}

\begin{proposition}
The measured relative entropy obeys the data-processing inequality; i.e., for
every state $\rho$, positive semi-definite operator $\sigma$, and quantum
channel $\mathcal{N}$, the following inequality holds:
\begin{equation}
D^{M}(\rho\Vert\sigma)\geq D^{M}(\mathcal{N}(\rho)\Vert\mathcal{N}(\sigma)).
\end{equation}

\end{proposition}

\subsection{Variational formulas for the measured relative entropy of states}

For a state $\rho$ and a positive semi-definite operator $\sigma$, the
following variational formulas for the measured relative entropy are known:
\begin{align}
D^{M}(\rho\Vert\sigma)  &  =\sup_{\omega>0}\left\{  \operatorname{Tr}[\left(
\ln\omega\right)  \rho]-\ln\operatorname{Tr}[\omega\sigma]\right\} \\
&  =\sup_{\omega>0}\left\{  \operatorname{Tr}[\left(  \ln\omega\right)
\rho]-\operatorname{Tr}[\omega\sigma]+1\right\}  .
\label{eq:BFT-var-form-meas-rel-ent}
\end{align}
The first was established in~\cite[Eq.~(11.69)]{petz2007quantum} and
\cite[Lemma~1 and Theorem~2]{Berta2015OnEntropies}, while the second was
established in~\cite[Lemma~1 and Theorem~2]{Berta2015OnEntropies}. Due to the
fact that the function $\omega\mapsto\ln\omega$ is operator concave, it
follows that the function $\omega\mapsto\operatorname{Tr}[\left(  \ln
\omega\right)  \rho]-\operatorname{Tr}[\omega\sigma]+1$ is concave, which is a
notable feature of the variational representation in
\eqref{eq:BFT-var-form-meas-rel-ent} that has further implications discussed
in the next section.

\subsection{Optimizing the measured relative entropy of states}

In this section, we observe that there is an efficient algorithm for computing
the measured relative entropy, which makes use of the observation
in~\eqref{eq:log-perspective-identity-simple}\ and the fact recalled in
Section~\ref{sec:op-connection}.

\begin{proposition}
\label{prop:meas-rel-ent-states-var-rep-perspective}Let $\rho$ be a state and
$\sigma$ a positive semi-definite operator. Then
\begin{equation}
D^{M}(\rho\Vert\sigma)=\sup_{\omega>0, \theta\in\mathbb{H}}\left\{  \operatorname{Tr}
[\theta\rho]-\operatorname{Tr}[\omega\sigma]+1:\theta\leq P_{\ln}
(I,\omega)\right\}  , \label{eq:meas-SDP}
\end{equation}
where $P_{\ln}$ is defined in~\eqref{eq:log-perspect}.

Furthermore, when $\rho$ and $\sigma$ are $d\times d$ matrices, the quantity
$D^{M}(\rho\Vert\sigma)$ can be efficiently calculated by means of a
semi-definite program up to an additive error $\varepsilon$, by means of $O(\sqrt
{\ln(1/\varepsilon)})$ linear matrix inequalities, each of size $2d\times2d$. Specifically, there exist $m,k\in \mathbb{N}$ such that $m+k=O(\sqrt
{\ln(1/\varepsilon)})$ and the following inequality holds:
\begin{equation}
    \left | D^M(\rho\|\sigma)-D^M_{m,k}(\rho\|\sigma)\right|\leq \varepsilon,
    \label{eq:rmk-log-approx}
\end{equation}
where
\begin{equation}
D_{m,k}^{M}(\rho\|\sigma)\coloneqq\sup_{\substack{\omega>0,\theta\in\mathbb{H},\\
T_{1},\ldots,T_{m}\in\mathbb{H},\\
Z_{0},\ldots,Z_{k}\in\mathbb{H}
}
}\left\{ \begin{array}{c}
\Tr\!\left[\theta\rho\right]-\Tr\!\left[\omega\sigma\right]+1:\\
Z_{0}=\omega,\,\sum_{j=1}^{m}w_{j}T_{j}=2^{-k}\theta,\\
\left\{ \begin{bmatrix}Z_{i} & Z_{i+1}\\
Z_{i+1} & I
\end{bmatrix}\geq0\right\} _{i=0}^{k-1},\\
\left\{ \begin{bmatrix}Z_{k}-I-T_{j} & -\sqrt{t_{j}}T_{j}\\
-\sqrt{t_{j}}T_{j} & I-t_{j}T_{j}
\end{bmatrix}\geq0\right\} _{j=1}^{m}
\end{array}\right\} \label{eq:measured-rel-ent-approx-m-k}
\end{equation}
and, for all $j\in\left\{ 1,\ldots,m\right\} $, $w_{j}$ and $t_{j}$
are the weights and nodes, respectively, for the $m$-point Gauss--Legendre
quadrature on the interval $\left[0,1\right]$.
\end{proposition}

\begin{proof}
The formula in~\eqref{eq:meas-SDP} is a direct consequence of
\eqref{eq:BFT-var-form-meas-rel-ent} and the following identity:
\begin{equation}
\ln\omega=P_{\ln}(I,\omega), \label{eq:log-perspective-identity-simple}
\end{equation}
while noting that the optimal value of $\theta$ in~\eqref{eq:meas-SDP} is
equal to $P_{\ln}(I,\omega)$, the latter being a Hermitian operator.

As such, we have rewritten $D^{M}(\rho\Vert\sigma)$ in terms of the hypograph
of $P_{\ln}(I,\omega)$. By appealing to~\cite[Theorem~2 and Theorem~3]
{Fawzi2019}, it follows that $D^{M}(\rho\Vert\sigma)$ can be efficiently
calculated by means of a semi-definite program with the stated complexity. In more detail, following \cite[Eqs.~(9) and (11)]{Fawzi2019}, let us define the functions
$r_{m}(x)$ and $r_{m,k}(x)$ as follows:
\begin{align}
r_{m}(x) & \coloneqq\sum_{j=1}^{m}w_{j}\frac{x-1}{t_{j}\left(x-1\right)+1},\\
r_{m,k}(x) & \coloneqq2^{k}r_{m}(x^{1/2^{k}}),
\end{align}
where $w_{j}$ and $t_{j}$ are as stated above. We can then define
the following perspective function:
\begin{equation}
P_{r_{m,k}}(X,Y)\coloneqq X^{1/2}r_{m,k}(X^{-1/2}YX^{-1/2})X^{1/2},
\end{equation}
which is jointly concave in $X$ and $Y$ \cite[Proposition~3]{Fawzi2019},
and its hypograph cone as follows:
\begin{equation}
\hyp_{r_{m,k}}\coloneqq\left\{ \left(X,Y,T\right)\in\mathbb{P}_{>0}(\mathcal{H})\times\mathbb{P}_{>0}(\mathcal{H})\times\mathbb{H}(\mathcal{H}):P_{r_{m,k}}(X,Y)\geq T\right\} .
\label{eq:hypo-r-mk-def}
\end{equation}
Then \cite[Theorem~3]{Fawzi2019} states that $\left(X,Y,T\right)\in\hyp_{r_{m,k}}$
if and only if there exist $T_{1},\ldots,T_{m},Z_{0},\ldots,Z_{k}\in\mathbb{H}$
such that
\begin{align}
Z_{0} & =Y, \label{eq:SDP-rep-hypo-r-mk-1}\\
\sum_{j=1}^{m}w_{j}T_{j} & =2^{-k}T,\\
\begin{bmatrix}Z_{i} & Z_{i+1}\\
Z_{i+1} & X
\end{bmatrix} & \geq0\qquad\forall i\in\left\{ 0,\ldots,k-1\right\} ,\\
\begin{bmatrix}Z_{k}-X-T_{j} & -\sqrt{t_{j}}T_{j}\\
-\sqrt{t_{j}}T_{j} & X-t_{j}T_{j}
\end{bmatrix} & \geq0\qquad\forall j\in\left\{ 1,\ldots,m\right\} .
\label{eq:SDP-rep-hypo-r-mk-last}
\end{align}
As such, the quantity $D_{m,k}^{M}(\rho\|\sigma)$, as defined in
\eqref{eq:measured-rel-ent-approx-m-k}, is equal to
\begin{equation}
D_{m,k}^{M}(\rho\|\sigma)=\sup_{\omega>0,\theta\in\mathbb{H}}\left\{ \Tr\!\left[\theta\rho\right]-\Tr\!\left[\omega\sigma\right]+1:\theta\leq P_{r_{m,k}}(I,\omega)\right\} .
\end{equation}
To arrive at the claim in~\eqref{eq:rmk-log-approx}, we invoke Lemma~\ref{lem:ln-hypograph-approx},
which is a straightforward matrix generalization of \cite[Theorem~2]{Fawzi2019}.
Specifically, by invoking the first part of Lemma~\ref{lem:ln-hypograph-approx},
we conclude that
\begin{equation}
\theta\leq P_{\ln}(I,\omega)\quad\implies\quad\theta-\varepsilon I\leq P_{r_{m,k}}(I,\omega),
\end{equation}
so that
\begin{align}
D^{M}(\rho\|\sigma) & =\sup_{\omega>0,\theta\in\mathbb{H}}\left\{ \Tr\!\left[\theta\rho\right]-\Tr\!\left[\omega\sigma\right]+1:\theta\leq P_{\ln}(I,\omega)\right\} \\
 & \leq\sup_{\omega>0,\theta\in\mathbb{H}}\left\{ \Tr\!\left[\theta\rho\right]-\Tr\!\left[\omega\sigma\right]+1:\theta-\varepsilon I\leq P_{r_{m,k}}(I,\omega)\right\} \\
 & =\sup_{\omega>0,\theta'\in\mathbb{H}}\left\{ \Tr\!\left[\left(\theta'+\varepsilon I\right)\rho\right]-\Tr\!\left[\omega\sigma\right]+1:\theta'\leq P_{r_{m,k}}(I,\omega)\right\} \\
 & =\sup_{\omega>0,\theta'\in\mathbb{H}}\left\{ \Tr\!\left[\theta'\rho\right]-\Tr\!\left[\omega\sigma\right]+1:\theta'\leq P_{r_{m,k}}(I,\omega)\right\} +\varepsilon\\
 & =D_{m,k}^{M}(\rho\|\sigma)+\varepsilon.\label{eq:approx-rmk-ln-1}
\end{align}
The second equality follows from the substitution $\theta'=\theta-\varepsilon I$
and the penultimate equality follows because $\Tr\!\left[\left(\theta'+\varepsilon I\right)\rho\right]=\Tr\!\left[\theta'\rho\right]+\varepsilon$.
Now by invoking the second part of Lemma~\ref{lem:ln-hypograph-approx},
we conclude that
\begin{equation}
\theta\leq P_{r_{m,k}}(I,\omega)\quad\implies\quad\theta-\varepsilon I\leq P_{\ln}(I,\omega),
\end{equation}
so that
\begin{align}
D_{m,k}^{M}(\rho\|\sigma) & =\sup_{\omega>0,\theta\in\mathbb{H}}\left\{ \Tr\!\left[\theta\rho\right]-\Tr\!\left[\omega\sigma\right]+1:\theta\leq P_{r_{m,k}}(I,\omega)\right\} \\
 & \leq\sup_{\omega>0,\theta\in\mathbb{H}}\left\{ \Tr\!\left[\theta\rho\right]-\Tr\!\left[\omega\sigma\right]+1:\theta-\varepsilon I\leq P_{\ln}(I,\omega)\right\} \\
 & =\sup_{\omega>0,\theta'\in\mathbb{H}}\left\{ \Tr\!\left[\left(\theta'+\varepsilon I\right)\rho\right]-\Tr\!\left[\omega\sigma\right]+1:\theta'\leq P_{\ln}(I,\omega)\right\} \\
 & =\sup_{\omega>0,\theta'\in\mathbb{H}}\left\{ \Tr\!\left[\theta'\rho\right]-\Tr\!\left[\omega\sigma\right]+1:\theta'\leq P_{\ln}(I,\omega)\right\} +\varepsilon\\
 & =D^{M}(\rho\|\sigma)+\varepsilon.\label{eq:approx-rmk-ln-2}
\end{align}
Combining~\eqref{eq:approx-rmk-ln-1} and~\eqref{eq:approx-rmk-ln-2},
we conclude~\eqref{eq:rmk-log-approx}.
\end{proof}

\begin{lemma}
\label{lem:ln-hypograph-approx}Let $a>1$ and $\varepsilon>0$. Then
there exist $m,k\in\mathbb{N}$ with $m+k=O(\sqrt{\ln(1/\varepsilon)})$
such that
\begin{enumerate}
\item If $0\leq a^{-1}Y\leq X\leq aY$ and $\left(X,Y,T\right)\in\hyp_{\ln}$,
then $\left(X,Y,T-X\varepsilon\right)\in\hyp_{r_{m,k}}.$
\item If $0\leq a^{-1}Y\leq X\leq aY$ and $\left(X,Y,T\right)\in\hyp_{r_{m,k}}$,
then $\left(X,Y,T-X\varepsilon\right)\in\hyp_{\ln}.$
\end{enumerate}
\end{lemma}

\begin{proof}
The proof follows the approach given in the proof of
\cite[Theorem 6]{Fawzi2019}. Consider that
\begin{align}
0 & \leq a^{-1}Y\leq X\leq aY \label{eq:hyp-approx-step-0}\\
\implies\qquad0 & \leq a^{-1}I\leq Y^{-1/2}XY^{-1/2}\leq aI\label{eq:hyp-approx-step-1}\\
\implies\qquad0 & \leq a^{-1}I\leq Y^{1/2}X^{-1}Y^{1/2}\leq aI\label{eq:hyp-approx-step-2}\\
\implies\qquad0 & \leq a^{-1}I\leq X^{-1/2}YX^{-1/2}\leq aI,\label{eq:hyp-approx-step-3}
\end{align}
where~\eqref{eq:hyp-approx-step-1} follows by left and right multiplying~\eqref{eq:hyp-approx-step-0} by $Y^{-1/2}$,~\eqref{eq:hyp-approx-step-2} follows from operator
antimonotonicity of the function $x\mapsto x^{-1}$, specifically,
\begin{align}
a^{-1}I\leq Y^{-1/2}XY^{-1/2} & \implies Y^{1/2}X^{-1}Y^{1/2}\leq aI,\\
Y^{-1/2}XY^{-1/2} \leq aI & \implies a^{-1}I\leq Y^{1/2}X^{-1}Y^{1/2},
\end{align}
and~\eqref{eq:hyp-approx-step-3} follows because $Y^{1/2}X^{-1}Y^{1/2}$
and $X^{-1/2}YX^{-1/2}$ have the same spectrum. By  \cite[Theorem~1]{Fawzi2019},
there exist $m,k\in\mathbb{N}$ with $m+k=O(\sqrt{\ln(1/\varepsilon)})$
such that
\begin{equation}
\left|r_{m,k}(z)-\ln(z)\right|\leq\varepsilon,
\end{equation}
for all $z\in\left[a^{-1},a\right]$. Then applying this to~\eqref{eq:hyp-approx-step-3}
gives that
\begin{equation}
r_{m,k}(X^{-1/2}YX^{-1/2})-\ln(X^{-1/2}YX^{-1/2})\geq-\varepsilon I,
\end{equation}
which implies that
\begin{equation}
X^{1/2}r_{m,k}(X^{-1/2}YX^{-1/2})X^{1/2}-X^{1/2}\ln(X^{-1/2}YX^{-1/2})X^{1/2}\geq-\varepsilon X.\label{eq:hyp-approx-step-4}
\end{equation}
By assumption, $\left(X,Y,T\right)\in\hyp_{\ln}$, so that
\begin{equation}
X^{1/2}\ln(X^{-1/2}YX^{-1/2})X^{1/2}\geq T.\label{eq:hyp-approx-step-5}
\end{equation}
Now adding~\eqref{eq:hyp-approx-step-4} and~\eqref{eq:hyp-approx-step-5}
gives
\begin{equation}
X^{1/2}r_{m,k}(X^{-1/2}YX^{-1/2})X^{1/2}\geq T-\varepsilon X,
\end{equation}
which is equivalent to $P_{r_{m,k}}(X,Y)\geq T-\varepsilon X$, thus
establishing the first claim. The proof of the second claim is similar.
\end{proof}

\begin{remark}
\label{rem:alt-approach}An alternative approach for computing the measured
relative entropy is to set $\alpha=1-2^{-\ell}$ for $\ell\in\mathbb{N}$,
similar to what was done in \cite[Lemma~9]{fang2021geometric}. By appealing to
Proposition~\ref{prop:opt-meas-renyi-states} and
\eqref{eq:bound-performance-renyi-approach}, it follows that, in order to
achieve an additive error $\varepsilon$ in computing the measured relative entropy, the
semi-definite program resulting from this approach requires  $O(\ln
(1/\varepsilon))$ linear matrix inequalities, each of size $2d\times2d$. When
compared to the performance of the approach from
Proposition~\ref{prop:meas-rel-ent-states-var-rep-perspective}, it is clear
that this latter approach is preferred because it requires only $O(\sqrt
{\ln(1/\varepsilon)})$ linear matrix inequalities in order to achieve the same error. 
\end{remark}

\section{Measured R\'{e}nyi relative entropy of channels}

\label{sec:measured-renyi-channels}

\subsection{Definition and basic properties}

\begin{definition}[Measured R\'{e}nyi channel divergence]
\label{def:measured-renyi-channels}
Given a quantum channel $\mathcal{N}_{A\rightarrow B}$, a completely positive
map $\mathcal{M}_{A\rightarrow B}$, a Hamiltonian $H_{A}$ (Hermitian operator
acting on system $A$), and an energy constraint $E\in\mathbb{R}$, the
energy-constrained measured R\'{e}nyi relative entropy of channels is defined
for all $\alpha\in\left(  0,1\right)  \cup\left(  1,\infty\right)  $ as
\begin{multline}
D_{\alpha,H,E}^{M}(\mathcal{N}\Vert\mathcal{M}
)\coloneqq\label{eq:EC-meas-Renyi-channel}\\
\sup_{\substack{d_{R^{\prime}}\in\mathbb{N},\\\rho_{R^{\prime}A}\in
\mathbb{D}(\mathcal{H}_{R^{\prime}A}\mathcal{)}}}\left\{  D_{\alpha}
^{M}(\mathcal{N}_{A\rightarrow B}(\rho_{R^{\prime}A})\Vert\mathcal{M}
_{A\rightarrow B}(\rho_{R^{\prime}A})):\operatorname{Tr}[H_{A}\rho_{A}]\leq
E\right\}  .
\end{multline}
In what follows, for brevity, we also refer to the quantity in
\eqref{eq:EC-meas-Renyi-channel} as the measured R\'{e}nyi channel divergence.
\end{definition}

In~\eqref{eq:EC-meas-Renyi-channel}\ above, the supremum is taken not only
over every bipartite state $\rho_{R^{\prime}A}$ but also over the reference
system $R^{\prime}$ with dimension $d_{R^{\prime}}$. We define the measured
R\'{e}nyi channel divergence in this general way in order to allow for all
physically feasible ways of processing a channel; indeed, one prepares a state
$\rho_{R^{\prime}A}$, sends it through either $\mathcal{N}_{A\rightarrow B}$
or $\mathcal{M}_{A\rightarrow B}$, and processes the systems $R^{\prime}B$
with a measurement in order to distinguish the maps $\mathcal{N}_{A\rightarrow
B}$ or $\mathcal{M}_{A\rightarrow B}$. We impose an energy constraint only on
the input system $A$, because this is the simplest and most minimal
modification of an unconstrained channel divergence and it allows for all
physically plausible, yet unconstrained, reference systems. Imposing the
energy constraint in this way furthermore has the benefit of leading to an
efficient algorithm for computing $D_{\alpha,H,E}^{M}$ by semi-definite
programming (see Proposition~\ref{prop:var-rep-measured-renyi-channels}). Let
us note that imposing an energy constraint only on the input system $A$ is
similar to the approach taken when defining the Shirokov--Winter
energy-constrained diamond norm~\cite{Shirokov2018,Winter2017} or more general
energy-constrained channel divergences~\cite[Eq.~(12.12)]{SWAT18}. One obtains
the unconstrained measured R\'{e}nyi relative entropy of channels by setting
$H_{A}=I_{A}$ and $E=1$, so that the \textquotedblleft energy
constraint\textquotedblright\ becomes redundant with the constraint that
$\rho_{R^{\prime}A}$ is a state.

By appealing to~\eqref{eq:quasi-renyi}, we can also write
\begin{equation}
D_{\alpha,H,E}^{M}(\mathcal{N}\Vert\mathcal{M})=\frac{1}{\alpha-1}\ln
Q_{\alpha,H,E}^{M}(\mathcal{N}\Vert\mathcal{M}),
\end{equation}
where
\begin{multline}
Q_{\alpha,H,E}^{M}(\mathcal{N}\Vert\mathcal{M})\coloneqq\\
\left\{
\begin{array}
[c]{cc}
\inf_{\substack{d_{R^{\prime}}\in\mathbb{N},\\\rho_{R^{\prime}A}\in
\mathbb{D}(\mathcal{H}_{R^{\prime}A}\mathcal{)},\\\operatorname{Tr}[H_{A}
\rho_{A}]\leq E}}Q_{\alpha}^{M}(\mathcal{N}_{A\rightarrow B}(\rho_{R^{\prime
}A})\Vert\mathcal{M}_{A\rightarrow B}(\rho_{R^{\prime}A})) & \text{for }
\alpha\in\left(  0,1\right) \\
\sup_{\substack{d_{R^{\prime}}\in\mathbb{N},\\\rho_{R^{\prime}A}\in
\mathbb{D}(\mathcal{H}_{R^{\prime}A}\mathcal{)},\\\operatorname{Tr}[H_{A}
\rho_{A}]\leq E}}Q_{\alpha}^{M}(\mathcal{N}_{A\rightarrow B}(\rho_{R^{\prime
}A})\Vert\mathcal{M}_{A\rightarrow B}(\rho_{R^{\prime}A})) & \text{for }
\alpha>1
\end{array}
\right.  .
\end{multline}

Although the optimization in~\eqref{eq:EC-meas-Renyi-channel} is defined to be
over an unbounded space, it is possible to simplify the task by employing
basic quantum information-theoretic reasoning. Indeed, as stated in
Proposition~\ref{prop:simplify-measured-channels} below, one can write
$D_{\alpha,H,E}^{M}(\mathcal{N}\Vert\mathcal{M})$ and $Q_{\alpha,H,E}
^{M}(\mathcal{N}\Vert\mathcal{M})$\ in terms of the Choi operators of
$\mathcal{N}_{A\rightarrow B}$ and $\mathcal{M}_{A\rightarrow B}$, defined as
\begin{equation}
\Gamma_{RB}^{\mathcal{N}}\coloneqq\mathcal{N}_{A\rightarrow B}(\Gamma
_{RA}),\qquad\Gamma_{RB}^{\mathcal{M}}\coloneqq\mathcal{M}_{A\rightarrow
B}(\Gamma_{RA}),
\end{equation}
where the maximally entangled operator $\Gamma_{RA}$ is defined as
\begin{equation}
\Gamma_{RA}\coloneqq\sum_{i,j}|i\rangle\!\langle j|_{R}\otimes|i\rangle
\!\langle j|_{A},
\end{equation}
so that the reference system $R$ is isomorphic to the channel input system $A$
(i.e., the corresponding Hilbert spaces $\mathcal{H}_{R}$ and $\mathcal{H}
_{A}$ are isomorphic, denoted by $\mathcal{H}_{R}\simeq\mathcal{H}_{A}$):

\begin{proposition}
\label{prop:simplify-measured-channels}Given a quantum channel $\mathcal{N}
_{A\rightarrow B}$, a completely positive map $\mathcal{M}_{A\rightarrow B}$,
a Hamiltonian $H_{A}$, and an energy constraint $E\in\mathbb{R}$, the measured
R\'{e}nyi channel divergence can be written as follows for all $\alpha
\in\left(  0,1\right)  \cup\left(  1,\infty\right)  $:
\begin{equation}
D_{\alpha,H,E}^{M}(\mathcal{N}\Vert\mathcal{M})=\sup_{\substack{\rho_{R}
\in\mathbb{D}(\mathcal{H}_{R}),\\\operatorname{Tr}[H_{A}\rho_{A}]\leq
E}}\left\{  D_{\alpha}^{M}(\rho_{R}^{1/2}\Gamma_{RB}^{\mathcal{N}}\rho
_{R}^{1/2}\Vert\rho_{R}^{1/2}\Gamma_{RB}^{\mathcal{M}}\rho_{R}^{1/2})\right\}
,
\end{equation}
where $\mathcal{H}_{R}\simeq\mathcal{H}_{A}$ and $\rho_{R}=\rho_{A}$.
Equivalently,
\begin{multline}
Q_{\alpha,H,E}^{M}(\mathcal{N}\Vert\mathcal{M})=\\
\left\{
\begin{array}
[c]{cc}
\inf_{\substack{\rho_{R}\in\mathbb{D}(\mathcal{H}_{R}\mathcal{)}
,\\\operatorname{Tr}[H_{A}\rho_{A}]\leq E}}Q_{\alpha}^{M}(\rho_{R}^{1/2}
\Gamma_{RB}^{\mathcal{N}}\rho_{R}^{1/2}\Vert\rho_{R}^{1/2}\Gamma
_{RB}^{\mathcal{M}}\rho_{R}^{1/2}) & \text{for }\alpha\in\left(  0,1\right) \\
\sup_{\substack{\rho_{R}\in\mathbb{D}(\mathcal{H}_{R}\mathcal{)}
,\\\operatorname{Tr}[H_{A}\rho_{A}]\leq E}}Q_{\alpha}^{M}(\rho_{R}^{1/2}
\Gamma_{RB}^{\mathcal{N}}\rho_{R}^{1/2}\Vert\rho_{R}^{1/2}\Gamma
_{RB}^{\mathcal{M}}\rho_{R}^{1/2}) & \text{for }\alpha>1
\end{array}
\right.
\end{multline}

\end{proposition}

\begin{proof}
See~\cite[Lemma~6]{Cooney2016} or~\cite[Proposition~7.82]
{khatri2024principles} for a detailed proof. The only difference with the
optimization above and those from ~\cite[Lemma~6]{Cooney2016}
and~\cite[Proposition~7.82]{khatri2024principles} is the additional energy
constraint $\operatorname{Tr}[H_{A}\rho_{A}]\leq E$, which follows because
$\rho_{R}^{1/2} \Gamma_{RA} \rho_{R}^{1/2}$, with $\rho_{R} = \rho_{A}$, is
the canonical purification of the state $\rho_{A}$.
\end{proof}

\subsection{Optimizing the measured R\'{e}nyi relative entropy of channels}

In this section, we observe in
Proposition~\ref{prop:var-rep-measured-renyi-channels}\ that there is a
variational representation of the measured R\'{e}nyi channel divergence in
terms of a linear objective function and the hypograph or epigraph of the
weighted geometric mean. From this observation, we conclude that there is an
efficient semi-definite optimization algorithm for computing the measured
R\'{e}nyi channel divergence, which makes use of the fact recalled in
Section~\ref{sec:hypo-epi}. Not only does this algorithm compute the optimal
value of $Q_{\alpha,H,E}^{M}(\mathcal{N}\Vert\mathcal{M})$ for all $\alpha
\in\left(  0,1\right)  \cup\left(  1,\infty\right)  $, but it also determines
an optimal input state and measurement that achieves $Q_{\alpha,H,E}
^{M}(\mathcal{N}\Vert\mathcal{M})$, which is of significant value for
applications. The proof of
Proposition~\ref{prop:var-rep-measured-renyi-channels} results from a
straightforward combination of Proposition~\ref{prop:opt-meas-renyi-states},
Proposition~\ref{prop:simplify-measured-channels}, and the transformer
equality in~\eqref{eq:KA-transformer-eq}.

\begin{proposition}
\label{prop:var-rep-measured-renyi-channels}Given is a quantum channel
$\mathcal{N}_{A\rightarrow B}$, a completely positive map $\mathcal{M}
_{A\rightarrow B}$, a Hamiltonian $H_{A}$, and an energy constraint
$E\in\mathbb{R}$. Let $\Gamma^{\mathcal{N}}$ and $\Gamma^{\mathcal{M}}$ be the
Choi operators of $\mathcal{N}_{A\rightarrow B}$ and $\mathcal{M}
_{A\rightarrow B}$, respectively. For $\alpha\in\left(  0,1/2\right)  $,
\begin{equation}
Q_{\alpha,H,E}^{M}(\mathcal{N}\Vert\mathcal{M})=\inf_{\Omega,\Theta,\rho
>0}\left\{
\begin{array}
[c]{c}
\alpha\operatorname{Tr}[\Omega\Gamma^{\mathcal{N}}]+\left(  1-\alpha\right)
\operatorname{Tr}[\Theta\Gamma^{\mathcal{M}}]:\\
\operatorname{Tr}[\rho]=1,\quad\operatorname{Tr}[H\rho]\leq E,\\
\Theta\geq G_{\frac{\alpha}{\alpha-1}}(\rho\otimes I,\Omega)
\end{array}
\right\}  ,
\end{equation}
for $\alpha\in\lbrack1/2,1)$,
\begin{equation}
Q_{\alpha,H,E}^{M}(\mathcal{N}\Vert\mathcal{M})=\inf_{\Omega,\Theta,\rho
>0}\left\{
\begin{array}
[c]{c}
\alpha\operatorname{Tr}[\Theta\Gamma^{\mathcal{N}}]+\left(  1-\alpha\right)
\operatorname{Tr}[\Omega\Gamma^{\mathcal{M}}]:\\
\operatorname{Tr}[\rho]=1,\quad\operatorname{Tr}[H\rho]\leq E,\\
\Theta\geq G_{1-\frac{1}{\alpha}}(\rho\otimes I,\Omega)
\end{array}
\right\}  ,\label{eq:var-exp-ch-a-1/2-1}
\end{equation}
and for $\alpha>1$,
\begin{equation}
Q_{\alpha,H,E}^{M}(\mathcal{N}\Vert\mathcal{M})=\sup_{\Omega,\Theta,\rho
>0}\left\{
\begin{array}
[c]{c}
\alpha\operatorname{Tr}[\Theta\Gamma^{\mathcal{N}}]+\left(  1-\alpha\right)
\operatorname{Tr}[\Omega\Gamma^{\mathcal{M}}]:\\
\operatorname{Tr}[\rho]=1,\quad\operatorname{Tr}[H\rho]\leq E,\\
\Theta\leq G_{1-\frac{1}{\alpha}}(\rho\otimes I,\Omega)
\end{array}
\right\}  .\label{eq:var-exp-ch-a-gt-1}
\end{equation}

For rational $\alpha\in\left(  0,1\right)  \cup\left(  1,\infty\right)  $, the
quantity $Q_{\alpha,H,E}^{M}(\mathcal{N}\Vert\mathcal{M})$ can be calculated
by means of a semi-definite program. More specifically, when $\Gamma
^{\mathcal{N}}$ and $\Gamma^{\mathcal{M}}$ are $d_{A}d_{B}\times d_{A}d_{B}$
matrices and $p$ and $q$ are relatively prime integers such that $\frac{p}
{q}=\frac{\alpha}{\alpha-1}$ for $\alpha\in\left(  0,1/2\right)  $ or
$\frac{p}{q}=1-\frac{1}{\alpha}$ for $\alpha\in\left[  1/2,1\right)
\cup(1,\infty)$, the semi-definite program requires $O(\log_{2}q)$ linear
matrix inequalities each of size $2d_{A}d_{B}\times2d_{A}d_{B}$.
\end{proposition}

\begin{proof}
This follows by combining Proposition~\ref{prop:opt-meas-renyi-states}\ and
Proposition~\ref{prop:simplify-measured-channels} and employing the
transformer equality in~\eqref{eq:KA-transformer-eq}. Let us consider the case
$\alpha\in\left(  0,1/2\right)  $:
\begin{align}
&  Q_{\alpha,H,E}^{M}(\mathcal{N}\Vert\mathcal{M})\nonumber\\
&  =\inf_{\substack{\rho_{R}\in\mathbb{D}(\mathcal{H}_{R}\mathcal{)}
,\\\operatorname{Tr}[H_{A}\rho_{A}]\leq E}}Q_{\alpha}^{M}(\rho_{R}^{1/2}
\Gamma_{RB}^{\mathcal{N}}\rho_{R}^{1/2}\Vert\rho_{R}^{1/2}\Gamma
_{RB}^{\mathcal{M}}\rho_{R}^{1/2})\label{eq:proof-channel-renyi-1}\\
&  =\inf_{\substack{\Omega^{\prime},\Theta^{\prime}>0,\\\rho_{R}\in
\mathbb{D}(\mathcal{H}_{R}\mathcal{)},\\\operatorname{Tr}[H_{A}\rho_{A}]\leq
E}}\left\{
\begin{array}
[c]{c}
\alpha\operatorname{Tr}[\Omega^{\prime}\rho_{R}^{1/2}\Gamma_{RB}^{\mathcal{N}
}\rho_{R}^{1/2}]+\left(  1-\alpha\right)  \operatorname{Tr}[\Theta^{\prime
}\rho_{R}^{1/2}\Gamma_{RB}^{\mathcal{M}}\rho_{R}^{1/2}]:\\
\Theta^{\prime}\geq G_{\frac{\alpha}{\alpha-1}}(I,\Omega^{\prime})
\end{array}
\right\}  \\
&  =\inf_{\substack{\Omega^{\prime},\Theta^{\prime},\rho_{R}>0,\\\operatorname{Tr}[\rho_{R}]=1\\\operatorname{Tr}[H_{A}\rho_{A}]\leq
E}}\left\{
\begin{array}
[c]{c}
\alpha\operatorname{Tr}[\rho_{R}^{1/2}\Omega^{\prime}\rho_{R}^{1/2}\Gamma
_{RB}^{\mathcal{N}}]+\left(  1-\alpha\right)  \operatorname{Tr}[\rho_{R}
^{1/2}\Theta^{\prime}\rho_{R}^{1/2}\Gamma_{RB}^{\mathcal{M}}]:\\
\Theta^{\prime}\geq G_{\frac{\alpha}{\alpha-1}}(I,\Omega^{\prime})
\end{array}
\right\}  \\
&  =\inf_{\substack{\Omega,\Theta,\rho_{R}>0,\\ \operatorname{Tr}[\rho
_{R}]=1\\\operatorname{Tr}[H_{A}\rho_{A}]\leq E}}\left\{
\begin{array}
[c]{c}
\alpha\operatorname{Tr}[\Omega\Gamma_{RB}^{\mathcal{N}}]+\left(
1-\alpha\right)  \operatorname{Tr}[\Theta\Gamma_{RB}^{\mathcal{M}}]:\\
\Theta\geq G_{\frac{\alpha}{\alpha-1}}(\rho\otimes I,\Omega)
\end{array}
\right\}  .\label{eq:proof-channel-renyi-last}
\end{align}
The first equality follows from
Proposition~\ref{prop:simplify-measured-channels}. The second equality follows
from Proposition~\ref{prop:opt-meas-renyi-states}. The third equality follows
from cyclicity of trace and the fact that the function
\begin{equation}
\rho_{R}\mapsto\alpha\operatorname{Tr}[\rho_{R}^{1/2}\Omega^{\prime}\rho
_{R}^{1/2}\Gamma_{RB}^{\mathcal{N}}]+\left(  1-\alpha\right)
\operatorname{Tr}[\rho_{R}^{1/2}\Theta^{\prime}\rho_{R}^{1/2}\Gamma
_{RB}^{\mathcal{M}}]
\end{equation}
is continuous in $\rho_{R}$, so that the optimization can be performed over
the set of positive definite density operators (dense in the set of all
density operators). The final equality follows from defining
\begin{equation}
\Omega\coloneqq\rho_{R}^{1/2}\Omega^{\prime}\rho_{R}^{1/2},\qquad
\Theta\coloneqq\rho_{R}^{1/2}\Theta^{\prime}\rho_{R}^{1/2}
,\label{eq:subst-defs}
\end{equation}
and noting that
\begin{equation}
\Omega^{\prime},\Theta^{\prime}>0\quad\Leftrightarrow\quad\Omega,\Theta>0,
\end{equation}
as well as
\begin{align}
&  \quad\Theta^{\prime}\geq G_{\frac{\alpha}{\alpha-1}}(I,\Omega^{\prime
})\nonumber\\
&  \Leftrightarrow\quad\left(  \rho^{1/2}\otimes I\right)  \Theta^{\prime
}\left(  \rho^{1/2}\otimes I\right)  \geq\left(  \rho^{1/2}\otimes I\right)
G_{\frac{\alpha}{\alpha-1}}(I,\Omega^{\prime})\left(  \rho^{1/2}\otimes
I\right)  \\
&  \Leftrightarrow\quad\Theta\geq G_{\frac{\alpha}{\alpha-1}}(\rho\otimes
I,\Omega),
\end{align}
with the final equality following from the definitions in
\eqref{eq:subst-defs}\ and the transformer equality in
\eqref{eq:KA-transformer-eq}. The proofs of~\eqref{eq:var-exp-ch-a-1/2-1} and
\eqref{eq:var-exp-ch-a-gt-1}\ follow similarly.

As such, we have rewritten $Q_{\alpha,H,E}^{M}(\mathcal{N}\Vert\mathcal{M})$
for $\alpha\in\left(  0,1/2\right)  $ in terms of hypograph of $G_{\frac
{\alpha}{\alpha-1}}$, for $\alpha\in\lbrack1/2,1)$ in terms of the hypograph
of $G_{1-\frac{1}{\alpha}}$, and for $\alpha>1$ in terms of the epigraph of
$G_{1-\frac{1}{\alpha}}$. By appealing to~\cite[Theorem~3]{Fawzi2017}, it
follows that all of these quantities can be efficiently calculated for
rational $\alpha$ by means of semi-definite programming, with the stated complexity.
\end{proof}

\section{Measured relative entropy of channels}

\label{sec:measured-rel-ent-channels}

\subsection{Definition and basic properties}

\begin{definition}[Measured relative entropy of channels]
\label{def:measured-rel-ent-channels}
Given a quantum channel $\mathcal{N}_{A\rightarrow B}$, a completely positive
map $\mathcal{M}_{A\rightarrow B}$, a Hamiltonian $H_{A}$ (Hermitian operator
acting on system $A$), and an energy constraint $E\in\mathbb{R}$, the
energy-constrained measured relative entropy of channels is defined as
\begin{multline}
D_{H,E}^{M}(\mathcal{N}\Vert\mathcal{M}
)\coloneqq\label{eq:measured-rel-ent-channels-EC}\\
\sup_{\substack{d_{R^{\prime}}\in\mathbb{N},\\\rho_{R^{\prime}A}\in
\mathbb{D}(\mathcal{H}_{R^{\prime}A}\mathcal{)}}}\left\{  D^{M}(\mathcal{N}
_{A\rightarrow B}(\rho_{R^{\prime}A})\Vert\mathcal{M}_{A\rightarrow B}
(\rho_{R^{\prime}A})):\operatorname{Tr}[H_{A}\rho_{A}]\leq E\right\}  .
\end{multline}
\end{definition}

The motivation for this definition is the same as that given after~\eqref{eq:EC-meas-Renyi-channel}.

Similar to what was observed in
Proposition~\ref{prop:simplify-measured-channels}, although the optimization
in~\eqref{eq:measured-rel-ent-channels-EC} is defined to be over an unbounded
space, it is possible to simplify the optimization task as follows.

\begin{proposition}
\label{prop:simplify-measured-channels-standard}Given a quantum channel
$\mathcal{N}_{A\rightarrow B}$, a completely positive map $\mathcal{M}
_{A\rightarrow B}$, a Hamiltonian $H_{A}$, and an energy constraint
$E\in\mathbb{R}$, the measured relative entropy of channels can be written as
follows:
\begin{equation}
D_{H,E}^{M}(\mathcal{N}\Vert\mathcal{M})=\sup_{\substack{\rho_{R}\in
\mathbb{D}(\mathcal{H}_{R}),\\\operatorname{Tr}[H_{A}\rho_{A}]\leq E}}\left\{
D^{M}(\rho_{R}^{1/2}\Gamma_{RB}^{\mathcal{N}}\rho_{R}^{1/2}\Vert\rho_{R}
^{1/2}\Gamma_{RB}^{\mathcal{M}}\rho_{R}^{1/2})\right\}  ,
\end{equation}
where $\mathcal{H}_{R}\simeq\mathcal{H}_{A}$ and $\rho_{R}=\rho_{A}$.
\end{proposition}

\begin{proof}
The proof is the same as that given for
Proposition~\ref{prop:simplify-measured-channels}.
\end{proof}

\subsection{Optimizing the measured relative entropy of channels}

In this section, we observe in
Proposition~\ref{prop:var-rep-measured-div-channels}\ that there is a
variational representation of the measured relative entropy of channels in
terms of a linear objective function and the hypograph of the operator
connection of the logarithm. From this observation, we conclude that there is
an efficient semi-definite optimization algorithm for computing the measured
relative entropy of channels, which makes use of the fact recalled in
Section~\ref{sec:op-connection}. Not only does this algorithm compute the
optimal value of $D_{H,E}^{M}(\mathcal{N}\Vert\mathcal{M})$, but it also
determines an optimal input state and measurement that achieves $D_{H,E}
^{M}(\mathcal{N}\Vert\mathcal{M})$, which, as mentioned previously, is of
significant value for applications. The proof of
Proposition~\ref{prop:var-rep-measured-div-channels} results from a
straightforward combination of
Proposition~\ref{prop:meas-rel-ent-states-var-rep-perspective},
Proposition~\ref{prop:simplify-measured-channels-standard}, and the
transformer equality in~\eqref{eq:transformer-conn-eq}.

\begin{proposition}
\label{prop:var-rep-measured-div-channels}Given is a quantum channel
$\mathcal{N}_{A\rightarrow B}$, a completely positive map $\mathcal{M}
_{A\rightarrow B}$, a Hamiltonian $H_{A}$, and an energy constraint
$E\in\mathbb{R}$. Let $\Gamma^{\mathcal{N}}$ and $\Gamma^{\mathcal{M}}$ be the
Choi operators of $\mathcal{N}_{A\rightarrow B}$ and $\mathcal{M}
_{A\rightarrow B}$, respectively. Then
\begin{equation}
D_{H,E}^{M}(\mathcal{N}\Vert\mathcal{M})=\sup_{\Omega,\rho>0, \Theta\in \mathbb{H}}\left\{
\begin{array}
[c]{c}
\operatorname{Tr}[\Theta\Gamma^{\mathcal{N}}]-\operatorname{Tr}[\Omega
\Gamma^{\mathcal{M}}]+1:\\
\operatorname{Tr}[\rho]=1,\quad\operatorname{Tr}[H\rho]\leq E,\\
\Theta\leq P_{\ln}(\rho\otimes I,\Omega)
\end{array}
\right\}  ,
\end{equation}
where $P_{\ln}$ is defined in~\eqref{eq:log-perspect}.

Furthermore, when $\Gamma^{\mathcal{N}}$ and $\Gamma^{\mathcal{M}}$ are $d_{A}
d_{B}\times d_{A} d_{B}$ matrices, the quantity $D_{H,E}^{M}(\mathcal{N}
\Vert\mathcal{M})$ can be efficiently calculated by means of a semi-definite
program up to an additive error $\varepsilon$, by means of $O(\sqrt{\ln(1/\varepsilon
)})$ linear matrix inequalities, each of size $2d_{A} d_{B}\times2d_{A} d_{B}$. Specifically, there exist $m,k\in \mathbb{N}$ such that $m+k=O(\sqrt
{\ln(1/\varepsilon)})$ and the following inequality holds:
\begin{equation}
    \left | D_{H,E}^{M}(\mathcal{N}\Vert\mathcal{M})-D_{H,E,m,k}^{M}(\mathcal{N}\Vert\mathcal{M})\right|\leq \varepsilon,
    \label{eq:rmk-log-approx-channels}
\end{equation}
where
\begin{multline}
D_{H,E,m,k}^{M}(\mathcal{N}\|\mathcal{M})\coloneqq\label{eq:measured-rel-ent-approx-m-k-channels}\\
\sup_{\substack{\Omega,\rho>0,\Theta\in\mathbb{H},\\
T_{1},\ldots,T_{m}\in\mathbb{H},\\
Z_{0},\ldots,Z_{k}\in\mathbb{H}
}
}\left\{ \begin{array}{c}
\Tr\!\left[\Theta\Gamma^{\mathcal{N}}\right]-\Tr\!\left[\Omega\Gamma^{\mathcal{M}}\right]+1:\\
\Tr[\rho]=1,\:\Tr[H\rho]\leq E,\\
Z_{0}=\Omega,\,\sum_{j=1}^{m}w_{j}T_{j}=2^{-k}\Theta,\\
\left\{ \begin{bmatrix}Z_{i} & Z_{i+1}\\
Z_{i+1} & \rho\otimes I
\end{bmatrix}\geq0\right\} _{i=0}^{k-1},\\
\left\{ \begin{bmatrix}Z_{k}-\rho\otimes I-T_{j} & -\sqrt{t_{j}}T_{j}\\
-\sqrt{t_{j}}T_{j} & \rho\otimes I-t_{j}T_{j}
\end{bmatrix}\geq0\right\} _{j=1}^{m}
\end{array}\right\} 
\end{multline}
and, for all $j\in\left\{ 1,\ldots,m\right\} $, $w_{j}$ and $t_{j}$
are the weights and nodes, respectively, for the $m$-point Gauss--Legendre
quadrature on the interval $\left[0,1\right]$.
\end{proposition}

\begin{proof}
This follows similarly to the proof of
Proposition~\ref{prop:var-rep-measured-renyi-channels}, and we provide the proof for
completeness. Indeed, here we combine
Proposition~\ref{prop:meas-rel-ent-states-var-rep-perspective}\ and
Proposition~\ref{prop:simplify-measured-channels-standard} and employ the
transformer equality in~\eqref{eq:transformer-conn-eq}. Consider that
\begin{align}
&  D_{H,E}^{M}(\mathcal{N}\Vert\mathcal{M})\nonumber\\
&  =\sup_{_{\substack{\rho_{R}\in\mathbb{D}(\mathcal{H}_{R}\mathcal{)}
,\\\operatorname{Tr}[H_{A}\rho_{A}]\leq E}}}D^{M}(\rho_{R}^{1/2}\Gamma
_{RB}^{\mathcal{N}}\rho_{R}^{1/2}\Vert\rho_{R}^{1/2}\Gamma_{RB}^{\mathcal{M}
}\rho_{R}^{1/2})\label{eq:proof-channel-standard-1}\\
&  =\sup_{\substack{\Omega^{\prime}>0,\Theta^{\prime}\in \mathbb{H},\\\rho_{R}\in
\mathbb{D}(\mathcal{H}_{R}\mathcal{)},\\\operatorname{Tr}[H_{A}\rho_{A}]\leq
E}}\left\{
\begin{array}
[c]{c}
\operatorname{Tr}[\Theta^{\prime}\rho_{R}^{1/2}\Gamma_{RB}^{\mathcal{N}}
\rho_{R}^{1/2}]-\operatorname{Tr}[\Omega^{\prime}\rho_{R}^{1/2}\Gamma
_{RB}^{\mathcal{M}}\rho_{R}^{1/2}]+1:\\
\Theta^{\prime}\leq P_{\ln}(I,\Omega^{\prime})
\end{array}
\right\}  \\
&  =\sup_{\substack{\Omega^{\prime},\rho_{R}>0,\Theta^{\prime}\in \mathbb{H},\\\operatorname{Tr}[\rho_{R}]=1,\\\operatorname{Tr}[H_{A}\rho_{A}]\leq
E}}\left\{
\begin{array}
[c]{c}
\operatorname{Tr}[\rho_{R}^{1/2}\Theta^{\prime}\rho_{R}^{1/2}\Gamma
_{RB}^{\mathcal{N}}]-\operatorname{Tr}[\rho_{R}^{1/2}\Omega^{\prime}\rho
_{R}^{1/2}\Gamma_{RB}^{\mathcal{M}}]+1:\\
\Theta^{\prime}\leq P_{\ln}(I,\Omega^{\prime})
\end{array}
\right\}  \\
&  =\sup_{\substack{\Omega,\rho_{R}>0,\Theta\in \mathbb{H},\\\operatorname{Tr}[\rho
_{R}]=1,\\\operatorname{Tr}[H_{A}\rho_{A}]\leq E}}\left\{
\begin{array}
[c]{c}
\operatorname{Tr}[\Theta\Gamma_{RB}^{\mathcal{N}}]-\operatorname{Tr}
[\Omega\Gamma_{RB}^{\mathcal{M}}]+1:\\
\Theta\leq P_{\ln}(\rho\otimes I,\Omega)
\end{array}
\right\}  .\label{eq:proof-channel-standard-last}
\end{align}
The first equality follows from
Proposition~\ref{prop:simplify-measured-channels-standard}. The second
equality follows from
Proposition~\ref{prop:meas-rel-ent-states-var-rep-perspective}. The third
equality follows from cyclicity of trace and the fact that the function
\begin{equation}
\rho_{R}\mapsto\operatorname{Tr}[\rho_{R}^{1/2}\Theta^{\prime}\rho_{R}
^{1/2}\Gamma_{RB}^{\mathcal{N}}]-\operatorname{Tr}[\rho_{R}^{1/2}
\Omega^{\prime}\rho_{R}^{1/2}\Gamma_{RB}^{\mathcal{M}}]+1
\end{equation}
is continuous in $\rho_{R}$, so that the optimization can be performed over
the set of positive definite density operators (dense in the set of all
density operators). The final equality follows from defining
\begin{equation}
\Omega\coloneqq\rho_{R}^{1/2}\Omega^{\prime}\rho_{R}^{1/2},\qquad
\Theta\coloneqq\rho_{R}^{1/2}\Theta^{\prime}\rho_{R}^{1/2}
,\label{eq:subst-defs-rel-ent}
\end{equation}
and noting that
\begin{equation}
\Omega^{\prime}>0, \quad \Theta^{\prime}\in \mathbb{H}\quad\Leftrightarrow\quad\Omega>0,\quad \Theta\in \mathbb{H},
\end{equation}
as well as
\begin{align}
&  \quad\Theta^{\prime}\leq P_{\ln}(I,\Omega^{\prime})\nonumber\\
&  \Leftrightarrow\quad\left(  \rho^{1/2}\otimes I\right)  \Theta^{\prime
}\left(  \rho^{1/2}\otimes I\right)  \leq\left(  \rho^{1/2}\otimes I\right)
P_{\ln}(I,\Omega^{\prime})\left(  \rho^{1/2}\otimes I\right)  \\
&  \Leftrightarrow\quad\Theta\leq P_{\ln}(\rho\otimes I,\Omega),
\end{align}
with the final equality following from the definitions in
\eqref{eq:subst-defs-rel-ent}\ and the transformer equality in~\eqref{eq:transformer-conn-eq}.

As such, we have rewritten $D_{H,E}^{M}(\mathcal{N}\Vert\mathcal{M})$ in terms
of the hypograph of $P_{\ln}$. By appealing to~\cite[Theorem~2 and
Theorem~3]{Fawzi2019}, it follows that this quantity can be efficiently
calculated by means of semi-definite programming, with the complexity stated above. In more detail, following the proof of Proposition~\ref{prop:meas-rel-ent-states-var-rep-perspective} closely and recalling
the definition of $\hyp_{r_{m,k}}$ and its semi-definite representation
in~\eqref{eq:hypo-r-mk-def}--\eqref{eq:SDP-rep-hypo-r-mk-last}, observe that
\begin{equation}
D_{H,E,m,k}^{M}(\mathcal{N}\|\mathcal{M})=\sup_{\Omega,\rho>0,\Theta\in\mathbb{H}}\left\{ \begin{array}{c}
\Tr\!\left[\Theta\Gamma^{\mathcal{N}}\right]-\Tr\!\left[\Omega\Gamma^{\mathcal{M}}\right]+1:\\
\Tr[\rho]=1,\:\Tr[H\rho]\leq E,\\
\Theta\leq P_{r_{m,k}}(\rho\otimes I,\Omega)
\end{array}\right\} .
\end{equation}
By invoking the first part of Lemma~\ref{lem:ln-hypograph-approx},
we conclude that
\begin{equation}
\Theta\leq P_{\ln}(\rho\otimes I,\Omega)\quad\implies\quad\Theta-\varepsilon\left(\rho\otimes I\right)\leq P_{r_{m,k}}(\rho\otimes I,\Omega),
\end{equation}
so that
\begin{align}
 & D_{H,E}^{M}(\mathcal{N}\|\mathcal{M})\nonumber \\
 & =\sup_{\Omega,\rho>0,\Theta\in\mathbb{H}}\left\{ \begin{array}{c}
\Tr\!\left[\Theta\Gamma^{\mathcal{N}}\right]-\Tr\!\left[\Omega\Gamma^{\mathcal{M}}\right]+1:\\
\Tr[\rho]=1,\:\Tr[H\rho]\leq E,\\
\Theta\leq P_{\ln}(\rho\otimes I,\Omega)
\end{array}\right\} \\
 & \leq\sup_{\Omega,\rho>0,\Theta\in\mathbb{H}}\left\{ \begin{array}{c}
\Tr\!\left[\Theta\Gamma^{\mathcal{N}}\right]-\Tr\!\left[\Omega\Gamma^{\mathcal{M}}\right]+1:\\
\Tr[\rho]=1,\:\Tr[H\rho]\leq E,\\
\Theta-\varepsilon\left(\rho\otimes I\right)\leq P_{r_{m,k}}(\rho\otimes I,\Omega)
\end{array}\right\} \\
 & =\sup_{\Omega,\rho>0,\Theta'\in\mathbb{H}}\left\{ \begin{array}{c}
\Tr\!\left[\left(\Theta'+\varepsilon\left(\rho\otimes I\right)\right)\Gamma^{\mathcal{N}}\right]-\Tr\!\left[\Omega\Gamma^{\mathcal{M}}\right]+1:\\
\Tr[\rho]=1,\:\Tr[H\rho]\leq E,\\
\Theta'\leq P_{r_{m,k}}(\rho\otimes I,\Omega)
\end{array}\right\} \\
 & =\sup_{\Omega,\rho>0,\Theta'\in\mathbb{H}}\left\{ \begin{array}{c}
\Tr\!\left[\Theta'\Gamma^{\mathcal{N}}\right]-\Tr\!\left[\Omega\Gamma^{\mathcal{M}}\right]+1:\\
\Tr[\rho]=1,\:\Tr[H\rho]\leq E,\\
\Theta'\leq P_{r_{m,k}}(\rho\otimes I,\Omega)
\end{array}\right\} +\varepsilon\\
 & =D_{H,E,m,k}^{M}(\mathcal{N}\|\mathcal{M})+\varepsilon.\label{eq:approx-rmk-ln-1-1}
\end{align}
The second equality follows from the substitution $\Theta'=\Theta-\varepsilon\left(\rho\otimes I\right)$
and the penultimate equality follows because
\begin{align}
\Tr\!\left[\left(\Theta'+\varepsilon\left(\rho\otimes I\right)\right)\Gamma^{\mathcal{N}}\right] & =\Tr\!\left[\Theta'\Gamma^{\mathcal{N}}\right]+\varepsilon\Tr\!\left[\left(\rho\otimes I\right)\Gamma^{\mathcal{N}}\right]\\
 & =\Tr\!\left[\Theta'\Gamma^{\mathcal{N}}\right]+\varepsilon,
\end{align}
given that $\Tr\!\left[\left(\rho\otimes I\right)\Gamma^{\mathcal{N}}\right]=\Tr\!\left[\mathcal{N}(\rho^T)\right]=1$, where the second equality follows from \cite[Eq.~(4.2.12)]{khatri2024principles}.
Now by invoking the second part of Lemma~\ref{lem:ln-hypograph-approx},
we conclude that
\begin{equation}
\Theta\leq P_{r_{m,k}}(\rho\otimes I,\Omega)\quad\implies\quad\Theta-\varepsilon\left(\rho\otimes I\right)\leq P_{\ln}(\rho\otimes I,\Omega),
\end{equation}
so that
\begin{align}
 & D_{H,E,m,k}^{M}(\mathcal{N}\|\mathcal{M})\nonumber \\
 & =\sup_{\Omega,\rho>0,\Theta\in\mathbb{H}}\left\{ \begin{array}{c}
\Tr\!\left[\Theta\Gamma^{\mathcal{N}}\right]-\Tr\!\left[\Omega\Gamma^{\mathcal{M}}\right]+1:\\
\Tr[\rho]=1,\:\Tr[H\rho]\leq E,\\
\Theta\leq P_{r_{m,k}}(\rho\otimes I,\Omega)
\end{array}\right\} \\
 & \leq\sup_{\Omega,\rho>0,\Theta\in\mathbb{H}}\left\{ \begin{array}{c}
\Tr\!\left[\Theta\Gamma^{\mathcal{N}}\right]-\Tr\!\left[\Omega\Gamma^{\mathcal{M}}\right]+1:\\
\Tr[\rho]=1,\:\Tr[H\rho]\leq E,\\
\Theta-\varepsilon\left(\rho\otimes I\right)\leq P_{\ln}(\rho\otimes I,\Omega)
\end{array}\right\} \\
 & =\sup_{\Omega,\rho>0,\Theta'\in\mathbb{H}}\left\{ \begin{array}{c}
\Tr\!\left[\left(\Theta'+\varepsilon\left(\rho\otimes I\right)\right)\Gamma^{\mathcal{N}}\right]-\Tr\!\left[\Omega\Gamma^{\mathcal{M}}\right]+1:\\
\Tr[\rho]=1,\:\Tr[H\rho]\leq E,\\
\Theta'\leq P_{\ln}(\rho\otimes I,\Omega)
\end{array}\right\} \\
 & =\sup_{\Omega,\rho>0,\Theta'\in\mathbb{H}}\left\{ \begin{array}{c}
\Tr\!\left[\Theta'\Gamma^{\mathcal{N}}\right]-\Tr\!\left[\Omega\Gamma^{\mathcal{M}}\right]+1:\\
\Tr[\rho]=1,\:\Tr[H\rho]\leq E,\\
\Theta'\leq P_{\ln}(\rho\otimes I,\Omega)
\end{array}\right\} +\varepsilon\\
 & =D_{H,E,m,k}^{M}(\mathcal{N}\|\mathcal{M})+\varepsilon.\label{eq:approx-rmk-ln-2-1}
\end{align}
Combining~\eqref{eq:approx-rmk-ln-1-1} and~\eqref{eq:approx-rmk-ln-2-1},
we conclude~\eqref{eq:rmk-log-approx-channels}.
\end{proof}

\section{Conclusion}

\label{sec:conclusion}

The main contributions of our paper are efficient semi-definite optimization
algorithms for computing measured relative entropies of quantum states and
channels. We did so by combining the results of
\cite{Berta2015OnEntropies,Fawzi2017,Fawzi2019} to obtain the result for
states and then we generalized these to quantum channels by using basic
properties of the weighted geometric mean and operator connection of the
logarithm. Our findings are of significant value for applications, in which
one wishes to find numerical characterizations of technologically feasible,
hybrid quantum-classical strategies for quantum hypothesis testing of states
and channels. 

Going forward from here, we note that further work in this direction could consider combining our findings here with those of \cite{rippchen2024locally}, the latter being about measured R\'enyi divergences under restricted forms of measurements. We also think it is a very interesting open question,
related to those from
\cite{Fawzi2017,fawzi2023optimalselfconcordantbarriersquantum}, to determine
semi-definite programs for various R\'{e}nyi relative entropies of quantum
channels, including those based on the sandwiched
\cite{muller2013quantum,wilde2014strong}\ and Petz--R\'{e}nyi \cite{P85,P86}
\ relative entropies. More generally, one could consider the same question for
$\alpha$-$z$ R\'{e}nyi relative entropies \cite{Audenaert2015}, and here we
think the variational formulas from \cite[Appendix~B]
{mosonyi2023continuitypropertiesquantumrenyi} could be useful in addressing
this question.

\bigskip

\textbf{Acknolwedgements}---We are grateful to Ludovico Lami for communicating
to us that the root fidelity of states $\rho$ and $\sigma$ can be written as
$\frac{1}{2}\inf_{Y,Z>0}\left\{  \operatorname{Tr}[Y\rho]+\operatorname{Tr}
[Z\sigma]:G_{1/2}(Y,Z)=I\right\}  $, which was a starting point for our work
here. We also thank James Saunderson for several email exchanges regarding
semi-definite optimization and the weighted geometric mean, related to
\cite{Fawzi2017}. ZH is supported by a Sydney Quantum Academy Postdoctoral
Fellowship and an ARC DECRA Fellowship (DE230100144) \textquotedblleft
Quantum-enabled super-resolution imaging\textquotedblright. She is also
grateful to the Cornell School of Electrical and Computer Engineering and the
Cornell Lab of Ornithology for hospitality during a June 2024 research visit.
MMW\ acknowledges support from the National Science Foundation under Grant
No.~2304816 and from Air Force Research Laboratory under agreement number FA8750-23-2-0031.

This material is based on research sponsored by Air Force Research Laboratory
under agreement number FA8750-23-2-0031. The U.S. Government is authorized to
reproduce and distribute reprints for Governmental purposes notwithstanding
any copyright notation thereon. The views and conclusions contained herein are
those of the authors and should not be interpreted as necessarily representing
the official policies or endorsements, either expressed or implied, of Air
Force Research Laboratory or the U.S. Government.

{\footnotesize
\bibliographystyle{alphaurl}
\bibliography{Ref-meas-SDP}
}

\appendix

\section{Alternative proof of variationals formulas for measured R\'enyi relative entropies}

\label{app:alt-proof-measured-renyi}

\begin{proposition}[\cite{Berta2015OnEntropies}]
\label{prop:alt-proof-var-measured-Renyi}For a state $\rho$ and a positive
semi-definite operator $\sigma$, the expression in
\eqref{eq:measured-Renyi-var-add} holds.
\end{proposition}

\begin{proof}
Let us begin with the case $\alpha\in\left(  0,1\right)  $. Let us write an
arbitrary $\omega>0$ in terms of a spectral decomposition as
\begin{equation}
\omega=\sum_{y}\omega_{y}|\phi_{y}\rangle\!\langle\phi_{y}
|,\label{eq:alt-proof-var-form-1}
\end{equation}
where $\omega_{y}>0$ for all $y$ and $\left\{  |\phi_{y}\rangle\right\}  _{y}$
is an orthonormal set. Consider that
\begin{align}
&  \alpha\operatorname{Tr}[\omega\rho]+\left(  1-\alpha\right)
\operatorname{Tr}[\omega^{\frac{\alpha}{\alpha-1}}\sigma]\nonumber\\
&  =\alpha\sum_{y}\omega_{y}\langle\phi_{y}|\rho|\phi_{y}\rangle+\left(
1-\alpha\right)  \sum_{y}\omega_{y}^{\frac{\alpha}{\alpha-1}}\langle\phi
_{y}|\sigma|\phi_{y}\rangle\\
&  =\sum_{y}\alpha\omega_{y}\langle\phi_{y}|\rho|\phi_{y}\rangle+\left(
1-\alpha\right)  \omega_{y}^{\frac{\alpha}{\alpha-1}}\langle\phi_{y}
|\sigma|\phi_{y}\rangle\\
&  \geq\sum_{y}\left[  \omega_{y}\langle\phi_{y}|\rho|\phi_{y}\rangle\right]
^{\alpha}\left[  \omega_{y}^{\frac{\alpha}{\alpha-1}}\langle\phi_{y}
|\sigma|\phi_{y}\rangle\right]  ^{1-\alpha}\label{eq:ineq-step-alt-proof}\\
&  =\sum_{y}\omega_{y}^{\alpha}\omega_{y}^{-\alpha}\langle\phi_{y}|\rho
|\phi_{y}\rangle^{\alpha}\langle\phi_{y}|\sigma|\phi_{y}\rangle^{1-\alpha}\\
&  =\sum_{y}\langle\phi_{y}|\rho|\phi_{y}\rangle^{\alpha}\langle\phi
_{y}|\sigma|\phi_{y}\rangle^{1-\alpha}.
\end{align}
The inequality follows from the inequality of weighted arithmetic and
geometric means (i.e., $\alpha b+\left(  1-\alpha\right)  c\geq b^{\alpha
}c^{1-\alpha}$ for all $b,c\geq0$ and $\alpha\in\left(  0,1\right)  $),
applied for all $y$. This inequality is saturated when the following condition
holds
\begin{equation}
\omega_{y}\langle\phi_{y}|\rho|\phi_{y}\rangle=\omega_{y}^{\frac{\alpha
}{\alpha-1}}\langle\phi_{y}|\sigma|\phi_{y}\rangle
.\label{eq:saturation-condition}
\end{equation}
That is, it is saturated when the two terms being averaged are equal.
Rearranging this, we find that saturation holds if
\begin{equation}
\omega_{y}=\left(  \frac{\langle\phi_{y}|\sigma|\phi_{y}\rangle}{\langle
\phi_{y}|\rho|\phi_{y}\rangle}\right)  ^{1-\alpha}.
\end{equation}
Thus, it follows that for every projective measurement $\left(  |\phi
_{y}\rangle\!\langle\phi_{y}|\right)  _{y}$, there exists $\omega>0$ such that
\begin{equation}
\alpha\operatorname{Tr}[\omega\rho]+\left(  1-\alpha\right)  \operatorname{Tr}
[\omega^{\frac{\alpha}{\alpha-1}}\sigma]=\sum_{y}\langle\phi_{y}|\rho|\phi
_{y}\rangle^{\alpha}\langle\phi_{y}|\sigma|\phi_{y}\rangle^{1-\alpha
}.\label{eq:alt-proof-var-form-last}
\end{equation}
As such we conclude the desired equality for all $\alpha\in\left(  0,1\right)
$:
\begin{equation}
\inf_{\omega>0}\left\{  \alpha\operatorname{Tr}[\omega\rho]+\left(
1-\alpha\right)  \operatorname{Tr}[\omega^{\frac{\alpha}{\alpha-1}}
\sigma]\right\}  =Q_{\alpha}^{P}(\rho\Vert\sigma).
\end{equation}
Then applying~\eqref{eq:proj-equals-measured} leads to the claim for
$\alpha\in\left(  0,1\right)  $.

The case $\alpha>1$ follows from a very similar proof, but instead makes use
of the following inequality: $\alpha b+\left(  1-\alpha\right)  c\leq
b^{\alpha}c^{1-\alpha}$, which holds for all $b\geq0$, $c>0$, and $\alpha>1$.
This inequality is a consequence of Bernoulli's inequality, which states that
$1+rx\leq\left(  1+x\right)  ^{r}$ holds for all $r\geq1$ and $x\geq-1$.
Indeed, consider that
\begin{align}
\alpha b+\left(  1-\alpha\right)  c\leq b^{\alpha}c^{1-\alpha}\quad &
\Longleftrightarrow\quad\alpha\left(  \frac{b}{c}\right)  +\left(
1-\alpha\right)  \leq\left(  \frac{b}{c}\right)  ^{\alpha}\\
&  \Longleftrightarrow\quad\alpha\left(  \frac{b}{c}-1\right)  +1\leq\left(
\frac{b}{c}-1+1\right)  ^{\alpha},
\end{align}
so that we choose $x=\frac{b}{c}-1\geq-1$ and $r=\alpha\geq1$ in Bernoulli's
inequality. So then we conclude the following for $\alpha>1$:
\begin{equation}
\sup_{\omega>0}\left\{  \alpha\operatorname{Tr}[\omega\rho]+\left(
1-\alpha\right)  \operatorname{Tr}[\omega^{\frac{\alpha}{\alpha-1}}
\sigma]\right\}  =Q_{\alpha}^{P}(\rho\Vert\sigma),
\end{equation}
by applying the same reasoning as in
\eqref{eq:alt-proof-var-form-1}--\eqref{eq:alt-proof-var-form-last}, but the
inequality in~\eqref{eq:ineq-step-alt-proof} goes in the opposite direction
for $\alpha>1$.
\end{proof}

\section{$\alpha \to 1$ limit of the measured R\'enyi relative entropy}

\label{app:alpha-1-limit}

\begin{proof}[Proof of Proposition~\ref{prop:alpha-1-limit-measured}]
This follows because
\begin{align}
\lim_{\alpha\rightarrow1^{-}}D_{\alpha}^{M}(\rho\Vert\sigma)  &  =\sup
_{\alpha\in\left(  0,1\right)  }D_{\alpha}^{M}(\rho\Vert\sigma)\\
&  =\sup_{\alpha\in\left(  0,1\right)  }\sup_{\mathcal{X},\left(  \Lambda
_{x}\right)  _{x\in\mathcal{X}}}D_{\alpha}(\left(  \operatorname{Tr}
[\Lambda_{x}\rho]\right)  _{x\in\mathcal{X}}\Vert\left(  \operatorname{Tr}
[\Lambda_{x}\sigma]\right)  _{x\in\mathcal{X}})\\
&  =\sup_{\mathcal{X},\left(  \Lambda_{x}\right)  _{x\in\mathcal{X}}}
\sup_{\alpha\in\left(  0,1\right)  }D_{\alpha}(\left(  \operatorname{Tr}
[\Lambda_{x}\rho]\right)  _{x\in\mathcal{X}}\Vert\left(  \operatorname{Tr}
[\Lambda_{x}\sigma]\right)  _{x\in\mathcal{X}})\\
&  =\sup_{\mathcal{X},\left(  \Lambda_{x}\right)  _{x\in\mathcal{X}}}D(\left(
\operatorname{Tr}[\Lambda_{x}\rho]\right)  _{x\in\mathcal{X}}\Vert\left(
\operatorname{Tr}[\Lambda_{x}\sigma]\right)  _{x\in\mathcal{X}})\\
&  =D^{M}(\rho\Vert\sigma),
\end{align}
as noted in~\cite[Lemma~2]{rippchen2024locally}, and because
\begin{align}
\lim_{\alpha\rightarrow1^{+}}D_{\alpha}^{M}(\rho\Vert\sigma)  &  =\inf
_{\alpha>1}D_{\alpha}^{M}(\rho\Vert\sigma)\\
&  =\inf_{\alpha>1}\sup_{\mathcal{X},\left(  \Lambda_{x}\right)
_{x\in\mathcal{X}}}D_{\alpha}(\left(  \operatorname{Tr}[\Lambda_{x}
\rho]\right)  _{x\in\mathcal{X}}\Vert\left(  \operatorname{Tr}[\Lambda
_{x}\sigma]\right)  _{x\in\mathcal{X}})\\
&  =\sup_{\mathcal{X},\left(  \Lambda_{x}\right)  _{x\in\mathcal{X}}}
\inf_{\alpha>1}D_{\alpha}(\left(  \operatorname{Tr}[\Lambda_{x}\rho]\right)
_{x\in\mathcal{X}}\Vert\left(  \operatorname{Tr}[\Lambda_{x}\sigma]\right)
_{x\in\mathcal{X}})\\
&  =\sup_{\mathcal{X},\left(  \Lambda_{x}\right)  _{x\in\mathcal{X}}}D(\left(
\operatorname{Tr}[\Lambda_{x}\rho]\right)  _{x\in\mathcal{X}}\Vert\left(
\operatorname{Tr}[\Lambda_{x}\sigma]\right)  _{x\in\mathcal{X}})\\
&  =D_{\alpha}^{M}(\rho\Vert\sigma).
\end{align}
The third equality above is non-trivial and follows from the fact that it
suffices to optimize $D_{\alpha}(\left(  \operatorname{Tr}[\Lambda_{x}
\rho]\right)  _{x\in\mathcal{X}}\Vert\left(  \operatorname{Tr}[\Lambda
_{x}\sigma]\right)  _{x\in\mathcal{X}})$ over POVMs with a finite number of
outcomes (a compact and convex set)~\cite[Theorem~4]{Berta2015OnEntropies},
that the relative entropy $D_{\alpha}$ is lower semi-continuous
\cite[Theorem~15]{Erven2014}, and an application of the Mosonyi--Hiai minimax
theorem~\cite[Lemma~A.2]{Mosonyi2011}.
\end{proof}

\end{document}